\documentclass{article}
\usepackage[utf8]{inputenc} 
\usepackage[T1]{fontenc}    
\usepackage{url}            
\usepackage{booktabs}       
\usepackage{amsfonts}       
\usepackage{nicefrac}       
\usepackage{microtype}      
\usepackage{xcolor}         
\usepackage{multirow}
\usepackage[left=2cm, right=2cm, top=1.5cm, bottom=1.5cm]{geometry}
\usepackage{graphicx} 
\usepackage{times}
\usepackage{soul}
\usepackage[hidelinks]{hyperref}
\usepackage[small]{caption}

\usepackage[switch]{lineno}

\usepackage{amsmath,amsthm,amssymb,mathrsfs,mathtools}
\usepackage{thm-restate}
\usepackage[ruled, linesnumbered, vlined]{algorithm2e}
\usepackage{tikz}
\usetikzlibrary{trees}
\usetikzlibrary{positioning}
\let\oldnl\nl
\newcommand{\nonl}{\renewcommand{\nl}{\let\nl\oldnl}}
\usepackage{cleveref} 
\usepackage{enumitem}

\newtheorem{theorem}{Theorem}[section]
\newtheorem{lemma}[theorem]{Lemma}
\theoremstyle{definition}
\newtheorem{definition}[theorem]{Definition}

\newtheorem{proposition}[theorem]{Proposition}

\newtheorem{observation}[theorem]{Observation}
\newcommand{\cI}{\mathcal{I}}   
\newcommand{\cM}{\mathcal{M}}   
\newcommand{\cCost}{\mathcal{C_{\mathcal{I}}}}
\newcommand{\cO}{\mathcal{O}}
\newcommand{\cGraph}{G}

\newcommand{\cItems}{{I}} 
\newcommand{\cItem}{i}  
\newcommand{\cG}{{I}}   
\newcommand{\cGCF}[2]{f_p^{#1}({#2})} 
\newcommand{\cPC}[3]{p_{#1}^{#2}{#3}} 
\newcommand{\cCF}[1]{f_p(#1)} 
\newcommand{\cEdge}[2]{({#1},{#2})}
\newcommand{\cFN}{\Gamma} 
\newcommand{\cProblem}{CCM}
\newcommand{\cProblemU}{CCMU}
\newtheorem{lp}[theorem]{LP}
\newtheorem{ip}[theorem]{ILP}
\newcommand{\cUtil}{u}
\newcommand{\cFNV}{V_{\Gamma}}
\newcommand{\cFNE}{E_{\Gamma}}

\author{Atasi Panda \\atasipanda@iisc.ac.in\\
Indian Institute of Science \and Harsh Sharma \\harshs.ug2023@cmi.ac.in\\
Chennai Mathematical Institute
\and Anand Louis \footnote{These two authors contributed equally}\\anandl@iisc.ac.in\\
Indian Institute of Science \and Prajakta Nimbhorkar \footnotemark[\value{footnote}] \\prajakta@cmi.ac.in\\
Chennai Mathematical Institute} 

\title{Group Fair Matchings using Convex Cost Functions}

\begin{document}

\maketitle

\begin{abstract}
We consider the problem of assigning items to platforms where each item has a utility associated with each of the platforms to which it can be assigned. Each platform has a soft constraint over the total number of items it serves, modeled via a convex cost function. Additionally, items are partitioned into groups, and each platform also incurs group-specific convex cost over the number of items from each group that can be assigned to the platform. These costs promote group fairness by penalizing imbalances, yielding a soft variation of fairness notions introduced in prior work, such as Restricted Dominance and Minority protection. Restricted Dominance enforces upper bounds on group representation, while Minority protection enforces lower bounds. Our approach replaces such hard constraints with cost-based penalties, allowing more flexible trade-offs. Our model also captures Nash Social Welfare kind of objective. 
 
The cost of an assignment is the sum of the values of all the cost functions across all the groups and platforms. The objective is to find an assignment that minimizes the cost while achieving a total utility that is at least a user-specified threshold. The main challenge lies in balancing the overall platform cost with group-specific costs, both governed by convex functions, while meeting the utility constraint. We present an efficient polynomial-time approximation algorithm, supported by theoretical guarantees and experimental evaluation. Our algorithm is based on techniques involving linear programming and network flows. We also provide an exact algorithm for a special case with uniform utilities and establish the hardness of the general problem when the groups can intersect arbitrarily. This work has applications in cloud computing, logistics, resource constrained machine learning deployment, federated learning, and network design, where resources must be allocated across platforms with diverse cost structures and diminishing returns.

\end{abstract}

\section{Introduction}

Bipartite graphs are a well-established framework for solving resource allocation and matching problems across diverse domains. In this paper, we work with bipartite graphs where the underlying partitions are referred to as items and platforms, and the goal is to compute a many-to-one matching on this graph, such that some constraints are satisfied. Each item has utilities for the platforms in its neighborhood, and each platform incurs a cost based on the set of items it serves. This cost is modeled using a convex function, which can capture various real-world phenomena such as increasing marginal cost due to operational overhead or resource constraints, or decreasing marginal cost up to a point, representing economies of scale until a critical mass is reached. Additionally, the items are divided into groups, and each platform also has convex cost functions associated with the number of items that it handles from each group. This allows us to encode preferences or constraints related to group fairness or diversity goals. For instance, serving too many items from a particular group may be penalized to avoid over-representation, while serving too few might incur inefficiencies due to under-representation or unmet minimum scale requirements.

The significance of group fairness constraints that enforce upper and lower bounds (quotas) on the number of items from each group assigned to a platform has been widely emphasized in literature \cite{vishnoi_fairness,luss_leximin_fair,devanur_ranking,CHRG16,halevi_fair_allocation,KMM15,BCZSK16,GroupIndividual}. For instance, in school choice, group fairness constraints can promote diversity among students assigned to each school based on attributes like ethnicity and socioeconomic background, as observed in practical implementations \cite{case-study}. Similarly, in project teams, group fairness constraints ensure the inclusion of experts from all required fields. The hospital-residents matching problem also benefits from group fairness constraints to ensure diverse specialties among assigned medical interns. These constraints thus achieve Restricted Dominance introduced in \cite{fair_clustering}, which asserts that the representation from any group on any platform does not exceed a user-specified cap, and Minority Protection \cite{fair_clustering}, which asserts that the representation from any group, among the items matched to any platform is at least a user-specified bound.
These principles restrict the over-representation or under-representation of any group on any platform by setting a threshold for each group's representation, ensuring balanced allocation across groups on each platform.  Fairness constraints with bounds on the number of items with each attribute are also studied in the context of ranking and multi-winner voting \cite{fair_ranking,fair_multivote}.

Our work departs from these strict quota systems by utilizing convex cost functions and introduces a softer version of this fairness condition, where the under- and/or over-representation of groups is penalized, thus allowing flexibility while enforcing fairness indirectly through cost minimization. This approach provides a model of fairness better suited for practical applications where rigid quotas may not be feasible or desirable. In many practical applications, specifying bounds optimally in advance is challenging, particularly when total allocations depend on dynamic factors such as demand, participation, or resource availability. Consequently, relaxing lower bounds or augmenting upper bounds is often considered to allow allocations to deviate from strict constraints \cite{one-sidedPreference_softquota}. This flexibility is crucial in scenarios where rigid adherence to bounds could exclude valid matches or lead to wastage. For instance, schools might relax regional quotas when there are insufficient applicants from a particular region, reallocating seats to other deserving candidates, thus dynamically increasing the quotas for those groups. The group-specific convex cost functions in our model ensure equitable representation among the groups in each platform through one of the following mechanisms: \\
1) \textbf{Over-representation Penalty (Soft Restricted Dominance)}: When the cost function is convex and increasing, assigning too many items from a group to a platform increases the cost, thereby limiting the group's representation.\\
2) \textbf{Under-representation penalty (Soft Minority Protection)}: When the cost function is convex and decreasing, it discourages assigning too few items from any group by imposing decreasing costs as the number of assigned items increases below a threshold level of representation. \\
3) \textbf{Balanced representation}: Achieve both $1)$ and $2)$ by having a convex cost function that changes direction at an inflection point. This simultaneously penalizes both under- and over-representation, encouraging group sizes within a preferred range.

 Nash Social Welfare (NSW) is a prominent welfare function in mathematical economics, which balances fairness and efficiency in resource distribution ~\cite{NSWFair}. The NSW objective seeks to maximize the geometric mean of individual (or group) utilities and strikes a balance between the utilitarian (arithmetic mean) and the egalitarian (minimum across values) criteria ~\cite{NSWFairSurvey}. This concept aligns with convex cost functions because these functions can naturally model diminishing returns for over-represented groups on platforms, ensuring that allocations are fair across groups. In resource allocation problems, NSW has been studied extensively in relation to competitive equilibria and is often used in fair division settings where both utility maximization and fairness are critical. By minimizing total costs with convex cost functions, the approach in this paper relates closely to achieving a balanced, equitable allocation akin to optimizing NSW in fairness-aware scenarios. Later, we formally show how NSW kind of objective can be captured by our framework in \Cref{prop:NSW}.
 
 Convex cost functions are often used to model scenarios where the marginal cost of serving additional items increases, such as in systems with congestion, limited capacity, or group-specific handling requirements. Our problem frequently arises in applications such as cloud computing, where servers have group-specific service costs (e.g., handling different types of tasks or workloads), or logistics, where transportation hubs face rising costs when serving more goods from certain categories. Our framework also naturally applies to resource-constrained Machine Learning deployment, where tasks or models must be assigned to servers with limited resources like compute or memory. The convex cost functions model the non-linear degradation or scaling penalties in overloaded systems. Group-specific costs enforce fairness across applications or user groups. The utility captures deployment performance goals such as responsiveness or latency for a model-server pair. 
 
 Our framework is also well suited to address data heterogeneity challenges in federated learning (FL). Typically in FL, client (item) participation is subject to server (platform) bandwidth. Our convex cost framework allows for control over group-level participation, where groups may correspond to data distributions, for example, datasets in different languages or in different regions. The use of group-specific convex costs enables the system to penalize skewed participation over time and incentivize more balanced inclusion of underrepresented clients. This leads to more robust global models that generalize better across non-uniform client distributions. The utility could represent the improvement in model performance from previous rounds, updated after each training round, before assigning clients to server slots for the next round.
 
\textit{In this paper, our goal is to efficiently compute a many-to-one matching that keeps both group-specific and total platform costs to a minimum while ensuring that the total utility is above a given threshold.}

\paragraph{Our Contributions:}
We present a framework for group fair resource allocation in bipartite matching where platforms incur convex costs based on both the total number of assigned items and the distribution of said items across groups. By replacing rigid group fairness constraints with convex penalties, our model enables flexible trade-offs between utility and fairness.
Our main contributions are:
\begin{itemize}
    \item {\bf Modeling group fairness via convex costs:} We propose convex cost functions as a way to encode soft constraints on group representation. This also captures objectives like Nash Social Welfare (see \Cref{prop:NSW}).
    \item {\bf Flow based algorithmic framework:} We reduce the group fair matching problem with convex costs to a Minimum Cost Flow problem under a total utility constraint when the groups are laminar.\footnote{A family of sets, say $S$, is laminar if, for every pair of sets $X, Y \in S$, one of the following holds: $X \subseteq Y$ or $Y \subseteq X$ or $X \cap Y = \phi$.} We design a multi-layer flow network that embeds group-specific and platform-level convex costs using piecewise linear approximations.
    \item {\bf Efficient approximation algorithm:} When the groups are disjoint or laminar, we develop an efficient Linear Programming based algorithm that produces an integral solution with provable runtime and approximation guarantees. This algorithm is exact when the utilities are uniform across all item-platform pairs.
    \item {\bf Hardness result for general groups:} We show that the problem becomes NP-hard for general (non-laminar) group structures via a reduction from the independent set problem.
\end{itemize}

\section{Our Problem and Results}\label{sec:prelims}
{\bf Convex Cost Matching with Utilities ($\cProblemU$):} We are given a connected bipartite graph $\cGraph = (\cItems, P, E)$ where $\cItems$ is the set of items, $P$ is the set of platforms, and $E \subseteq \cItems \times P$ is the set of edges representing the possible assignments of items to platforms. There are $n$ items, and $m$ platforms, $\delta(v)$ denotes all the edges with one end in the vertex $v$, $N(v)$ denotes the neighborhood of $v$, and $\Delta(v)$ denotes the degree of $v$. The items are divided into $\tau$, possibly overlapping groups, $\cG_1, \cG_2, \cdots, \cG_{\tau}$ where $\cG_j \subseteq \cItems$ represents the set of items in group $j$, such that $\cItems = \cG_1 \cup \cG_2 \cup \cdots \cup \cG_{\tau}$ and $\Delta_j(p)$ denotes the number of items in $N(p)$ from the group $j \in [\tau]$ i.e. $|N(p)\cap \cG_j|$. The items have different utilities associated with different platforms in their neighborhood, denoted by $\cUtil_e$ $\forall e \in E$. We are also given a lower bound, 
$\ell$, on the total utility that the resulting matching must satisfy. Each platform $p \in P$ has an associated convex cost function $\cCF{\sigma_p}$ where $\sigma_p$ denotes the total number of items matched to $p$. It also has a convex cost function for each of the $\tau$ groups, $\cGCF{1}{\nu_p^1}, \cGCF{2}{\nu_p^2},...,\cGCF{\tau}{\nu_p^{\tau}}$ where $\nu_p^j$ is the number of items matched to platform $p$ from $\cG_j$. When the groups are disjoint, $\sigma_p = \sum_{j=1}^{\tau}\nu_p^j$ otherwise $\sigma_p \leq \sum_{j=1}^{\tau}\nu_p^j$.

Let $\cI = (\cGraph, \cG_1, \cG_2, \cdots, \cG_{\tau}, \ell, \vec{\cUtil},F)$ be an instance of our problem where $F$ is a $m \times (\tau+1)$ matrix with $F_{pj}$ denoting the convex cost function associated with platform $p \in P$ and group $\cG_j$, $\cGCF{j}{.}$, except for $j=\tau+1$ when it denotes the convex cost function that depends on the total number of items matched to $p$, $\cCF{.}$. $\vec{u}$ is a $|E|$ vector with each entry corresponding to the utility of an edge in $E$. Note that the utilities, utility threshold, and the convex cost functions are all given to us as part of the input along with the input graph and item distribution across groups. We want to compute a matching with total utility at least $\ell$ while keeping the total cost to a minimum. Let $x_{e}$ be a variable associated with each edge, $e \in E$, such that for an edge $(i,p)$:
$$x_{ip} = \left\{
\begin{array}{ll}
      1 &\text{ if item $\cItem \in\cItems$ is assigned to platform $p \in P$} \\
      0 &\text{ otherwise} \\
\end{array} 
\right\}$$
\noindent Therefore, $\nu_p^j = \sum_{\begin{subarray}{l}
      \cItem \in N(p) \\
      i\in \cG_j
\end{subarray}}x_{ip}$, $\sigma_p = \sum_{\cItem \in N(p)}x_{ip}$ and the total cost of the platforms, say, $\cCost$, is
\begin{equation}\label{eq:Objective}
    \cCost = \sum_{p \in P}\left(\sum_{j=1}^{\tau}\cGCF{j}{\nu_p^j} + \cCF{\sigma_p}\right)
\end{equation}

\noindent Hence, the objective is to minimize $\cCost$ (\Cref{eq:Objective}) while satisfying the utility constraint. The following Integer Linear Program captures this problem.

\begin{ip}\label{ILP}
 \begin{align}
    &\min \cCost\\
    \text{such that} \quad &\sum_{e \in E}\cUtil_ex_e \geq \ell \label{eq:Constraint} \\
    &\sum_{e \in \delta(i)}x_e \leq 1 \quad \forall \cItem \in \cItems \\
    &x_e \in \{0,1\}  \quad \forall e \in E
 \end{align}
 \end{ip}


\subsection{Results}\label{subsec:results}
\noindent Let $OPT$ represent the minimum cost of a matching that satisfies the utility constraint, and $\nabla f(x)$ denote $f(x)-f(x-1)$ for any function $f(.)$. Our main contribution is a polynomial time algorithm that computes a matching with utility at least $\ell$ and the total cost close to $OPT$, when the groups are either disjoint or laminar, formally stated below.

\begin{theorem}\label{thm:disjoint}
    Given an instance $\cI = (\cGraph, \cG_1, \cG_2, \cdots, \cG_{\tau}, \ell, \vec{\cUtil},F)$ of the $\cProblemU$ problem with disjoint groups, there is a polynomial-time algorithm that computes a matching on $\cGraph$ such that the total utility is at least $\ell$ that is constraint \ref{eq:Constraint} is satisfied and the total cost of the platforms, $\cCost$ (\Cref{eq:Objective}), is at most $OPT + \max_p[\nabla\cCF{\sigma_p} + \max_j \nabla\cGCF{j}{\nu_p^j}] - \min_p[\nabla\cCF{\sigma_p} + \min_j \nabla\cGCF{j}{\nu_p^j}]$.
\end{theorem}
When the groups are overlapping but follow a laminar structure, let us say that the maximum number of groups an item can belong to is $d$. Due to the laminar structure, the groups form a depth-$d$ forest (refer to \Cref{fig:laminar_groups}). Let $j_r$ denote the group, $\cG_j, j \in [\tau]$ at level $r \in [d]$ in this forest. For laminar groups, we formally state our result below.
\begin{theorem}\label{thm:laminar}
    Given an instance $\cI = (\cGraph, \cG_1, \cG_2, \cdots, \cG_{\tau}, \ell, \vec{\cUtil},F)$ of the $\cProblemU$ problem, when the groups follow a laminar structure where the maximum number of groups an item can belong to is denoted by $d$, there is a polynomial-time algorithm that computes a matching on $\cGraph$ such that the total utility is at least $\ell$ that is constraint \ref{eq:Constraint} is satisfied and the total cost of the platforms, $\cCost$ (\Cref{eq:Objective}), is at most $OPT + \max_p[\nabla\cCF{\sigma_p} + \sum_{r=1}^{d}\max_{j_r} \nabla\cGCF{j_r}{\nu_p^{j_r}}] - \min_p[\nabla\cCF{\sigma_p} + \min_j \nabla\cGCF{j}{\nu_p^j}]$.
\end{theorem}

\noindent{\bf Hardness:} 
We give an NP-hardness result for the $\cProblemU$ problem when groups have a non-laminar structure using a reduction from the independent set problem inspired by a similar reduction in \cite{GroupFairMatching}. 
\noindent\begin{theorem}\label{thm:hardness}
The $\cProblemU$ problem is NP-hard when groups have a non-laminar structure. In fact, it cannot be approximated to any positive multiplicative factor.
\end{theorem}

We also contribute an exact algorithm for a special case of the $\cProblemU$ problem that has uniform utilities that is $\cUtil_e = q, \forall e  \in E$, for some $q \in \mathbb{R}$, and $\frac{\ell}{q}$ is an integer less than $|\cItems|$, when the groups are either disjoint or laminar. We formally state this result below.

\begin{restatable}{theorem}{CCM}\label{thm:CCM}
    Given an instance $\cI = (\cGraph, \cG_1, \cG_2, \cdots, \cG_{\tau}, \ell, q,F)$ of the $\cProblemU$ problem, when the groups are disjoint or laminar, and $\frac{\ell}{q}$ is a positive integer less than $|\cItems|$, there is a polynomial-time algorithm that computes a matching on $\cGraph$ such that the total utility is at least $\ell$ that is constraint \ref{eq:Constraint} is satisfied and the total cost of the platforms, $\cCost$ (\Cref{eq:Objective}), is minimized. 
\end{restatable}

Note that if $\frac{\ell}{q} \leq |\cItems|$ is a mixed fraction, we compute an integer matching that satisfies the utility constraint \ref{eq:Constraint} and the cost is at most $OPT + \max_p[\nabla\cCF{\lceil\sigma_p\rceil}(1 - (\lceil\sigma_p\rceil - \sigma_p))  + \max_j \nabla\cGCF{j}{\lceil\nu_p^j\rceil}(1 - (\lceil\nu_p^j\rceil - \nu_p^j))]$. If $\ell \leq q$, it is a trivial solution of one edge where we pick an edge incident on the platform $p = arg\min_{p}[\cCF{1} + \min_j \cGCF{j}{1}]$.  
If $\frac{\ell}{q} > |\cItems|$, it is immediate that no matching can satisfy \Cref{eq:Constraint}, as the required number of matches exceeds the total number of available items. The proofs of Theorems \ref{thm:laminar}, \ref{thm:hardness} and \ref{thm:CCM} can be found in \Cref{sec:laminar}, \Cref{sec:hardness} and \Cref{sec:CCM} respectively.

Finally, we establish the connection between our model and Nash Social Welfare with the following proposition:

\begin{proposition}\label{prop:NSW}
    Let $\nu_p^j$ denote the number of items from the group $\cG_j$ matched to platform $p$, and $\sigma_p$ denote the total number of items matched to platform $p$. Let $\tau_p$ be the number of groups in platform $p$'s neighborhood. Then, the following holds: \\
    1) Maximizing the Nash Social Welfare (NSW) across groups on platform $p$ corresponds to maximizing
    $(\Pi_{j}\nu_p^j)^{\frac{1}{\tau_p}}$.
    Its convex relaxation is the well-studied Eisenberg-Gale program, ~\cite{ApproxNSW} which maximizes $\frac{1}{\tau_p}\sum_{j=1}^{\tau_p}\log(\nu_p^j)$. \\
    2) Maximizing NSW for fair load distribution across platforms corresponds to maximizing $(\Pi_p\sigma_p)^{\frac{1}{m}}$
    whose convex relaxation is to maximize $\frac{1}{m} \sum_{p=1}^m \log(\sigma_p)$.\\
    Therefore, the overall relaxed objective becomes to maximize $\sum_{p=1}^m \left(\sum_{j=1}^{\tau}\frac{\log(\nu_p^j)}{\tau_p} + \frac{\log(\sigma_p)}{m}\right)$ or equivalently to minimize
    $$\sum_{p=1}^m \left(\sum_{j=1}^{\tau}\frac{-\log(\nu_p^j)}{\tau_p} + \frac{-\log(\sigma_p)}{m}\right).$$
    This expression matches the objective in \Cref{eq:Objective} when the platform cost function is set to $\frac{-\log(.)}{m}$, and every group-specific convex cost function for any platform, $p$, is set to $\frac{-\log(.)}{\tau_p}$.
    That is $\cCF{.} = \frac{-\log(.)}{m}$, $\forall p \in P$ and 
    $\cGCF{j}{.}= \frac{-\log(.)}{\tau_p}$, $\forall p \in P, j \in [\tau]$ in \Cref{eq:Objective}.
\end{proposition}

\section{Related Work}
Allocation problems are foundational in operations research and find applications in diverse fields such as resource allocation \cite{halabian_resourceallocation}, kidney exchange programs \cite{KidneyExchange}, school choice \cite{abdulkadiroglu2003}, candidate selection \cite{CandidateSelection}, summer internship programs \cite{SummerInternship}, and matching residents to hospitals \cite{Hospital-Resident}. These problems are often modeled as matching problems, where entities to be matched frequently belong to different groups. This has spurred significant research into bipartite matching under group fairness constraints. A comprehensive survey of developments in matching with constraints, including those based on regions, diversity, multi-dimensional capacities, and matroids, is provided in \cite{AzizBiroYokoo2022}.

Fairness constraints are integral to ensuring equitable outcomes in matching problems, and are captured in various forms like justified envy-freeness \cite{abdulkadiroglu2003}, proportional matching \cite{CandidateSelection}, and upper and lower bounds on the total representation of each group \cite{ClassifiedStableMatching,IsraeliGapYear}. In contexts addressing historical discrimination, vertical reservations (implemented as set-asides) and horizontal reservations (implemented as minimum guarantees or lower bounds) are prominent mechanisms used in India to safeguard disadvantaged groups, as detailed in \cite{IndiaReservation}. Moreover, groups in matching problems can be overlapping, that is one item could belong to multiple groups. \cite{GroupFairMatching} presents a polynomial-time algorithm with an approximation ratio of $\frac{1}{\Delta+1}$ where each item belongs to at most $\Delta$ laminar families of groups per platform, and \cite{Rank-MaximalAndPOpular} shows the NP-hardness of the problem without a laminar structure.
Both papers focus only on upper bounds, and \cite{GroupFairMatching} shows that their problem is NP-hard to approximate within a factor of $\cO\left(\frac{\log^2\Delta}{\Delta}\right)$. We show that this hardness result holds for our problem as well when the groups overlap arbitrarily (see \Cref{subsec:results} for details). \cite{DiversityMatching} primarily focuses on proportional diversity constraints, and \cite{GroupIndividual} studies individual fairness along with upper bounds and both upper and lower bounds for disjoint groups. All these works focus on satisfying the constraints while maximizing the matching size. In contrast, our work looks at matching in a way that the utility is above a certain threshold while minimizing the total cost imposed using convex cost functions. 

Traditionally, bounds or quotas, used to enforce fairness constraints, are fixed. However, recent studies have explored capacity expansion in applications such as school choice \cite{school_choice_expansion}, matching residents to hospitals \cite{hopital-resident_expansion}, and college admissions \cite{college_admissions_expansion}. \cite{Hospital-resident_stable} looks at optimally increasing the quotas of hospitals to ensure that a strongly stable matching exists in the hospital resident matching. They show that minimizing the maximum capacity increase for any hospital is NP-hard in general and provide a polynomial-time algorithm for minimizing the total capacity increase across all hospitals. However, even that becomes NP-hard when each hospital incurs a cost for each capacity increase, even if the costs are $0$ or $1$. 
\cite{EHLERS2014648}  studies the impact of
soft and hard diversity constraints (quotas) on the controlled school choice problem. \cite{Yenmez,ijcai2020softconstraints,overlappingSchoolChoice} are other works that look at soft diversity constraints in this problem, and \cite{overlappingSchoolChoice} considers the setting
in which each student has multiple types. For more details, we refer the reader to the survey paper \cite{Aziz_Biró_Yokoo_2022}. \cite{one-sidedPreference_softquota} looks at matching applicants to posts with one-sided preferences where the assignments can deviate from the range defined by the quotas. They study the computing of an assignment optimal w.r.t. preferences of
applicants with minimum deviation.

In \cite{one-sided_cost-based} and \cite{two-sided_cost-based}, cost-controlled quotas replace fixed quotas in one-sided and two-sided preference-based matching, respectively.  However, the problems considered in both these works are very different from ours. While these works also minimize costs, their notions of optimality (envy-freeness, rank-maximality) are preference-based and more aligned with individual fairness,
rather than group fairness. Additionally, these models do not consider group structures, and matching costs are linear rather than convex.
In contrast, we address group fairness in a cost-minimization framework by employing convex cost functions to enforce soft diversity constraints.

\section{Cost-Approximate Matching Algorithm for Disjoint Groups}\label{sec:algorithm}
We prove \Cref{thm:disjoint} in this section. Given an instance of the $\cProblemU$ problem, $\cI = (\cGraph, \cG_1, \cG_2, \cdots, \cG_{\tau}, \ell, \vec{\cUtil},F)$ with disjoint groups, the goal is to compute a matching that minimizes the total cost while ensuring that the total utility is at least $\ell$. To this end, we first construct a flow network from the given input instance (step \ref{step:construct} of \Cref{alg:main} and \Cref{fig:flow_network}) such that solving a Minimum Cost Flow problem on this network, subject to the utility constraint (Constraint~\ref{eq:Constraint}), yields a solution that can be directly translated into a valid matching in the original graph. We use \Cref{LP_int} to compute such a flow (step \ref{step:runLP}),  and if the solution is fractional, we invoke \Cref{alg:round} to round it to an integral flow. Finally, we map this flow to a bipartite matching that satisfies the utility constraint while possibly incurring a slightly higher cost than $OPT$. We show in \Cref{lem:frac_paths} that if an optimal solution of \Cref{LP_int} is fractional, it consists of at most two fractional paths. \Cref{alg:round} resolves this by rounding up the edges in the higher-utility path and rounding down those in the lower-utility path. We now delve into the details of our algorithm, starting with the network construction.

\subsection{Network Construction}\label{subsec:Flow}
We introduce two dummy nodes, $s$ and $t$, as source and sink, respectively, and construct a multi-layer directed bipartite multigraph, denoted by $\cFN$, from the input graph, $\cGraph$ (see \Cref{fig:flow_network}). Let the set of vertices and edges in $\cFN$ be denoted by $\cFNV$ and $\cFNE$ respectively. For any platform, $p \in P$, let $g(p)$ denote the set of groups with at least one item in the neighborhood of platform $p$, $N(p)$.
\newline
\noindent{\bf Nodes:} We create the following layers from the vertex sets $\cItems$ and $P$ ( \Cref{fig:flow_network}). \\
    1. {\bf Item} - This set is precisely $\cItems$, the collection of all items in the instance. \\
    2. {\bf Group} - For each platform, $p \in P$, we create a separate copy for every group $j \in g(p)$. 
    We denote the copy of a platform $p \in P$ corresponding to group $j \in g(p)$, by $\cPC{}{j}{}$. \\ 
     3. {\bf Platform} - In this layer, we introduce two levels, each containing a copy of every platform. For any platform $p\in P$, we denote its copies in the first and second levels as $\cPC{}{}{(l_1)}$ and $\cPC{}{}{(l_2)}$ respectively.

\begin{figure}[t]
\centering
\includegraphics[scale=0.2]{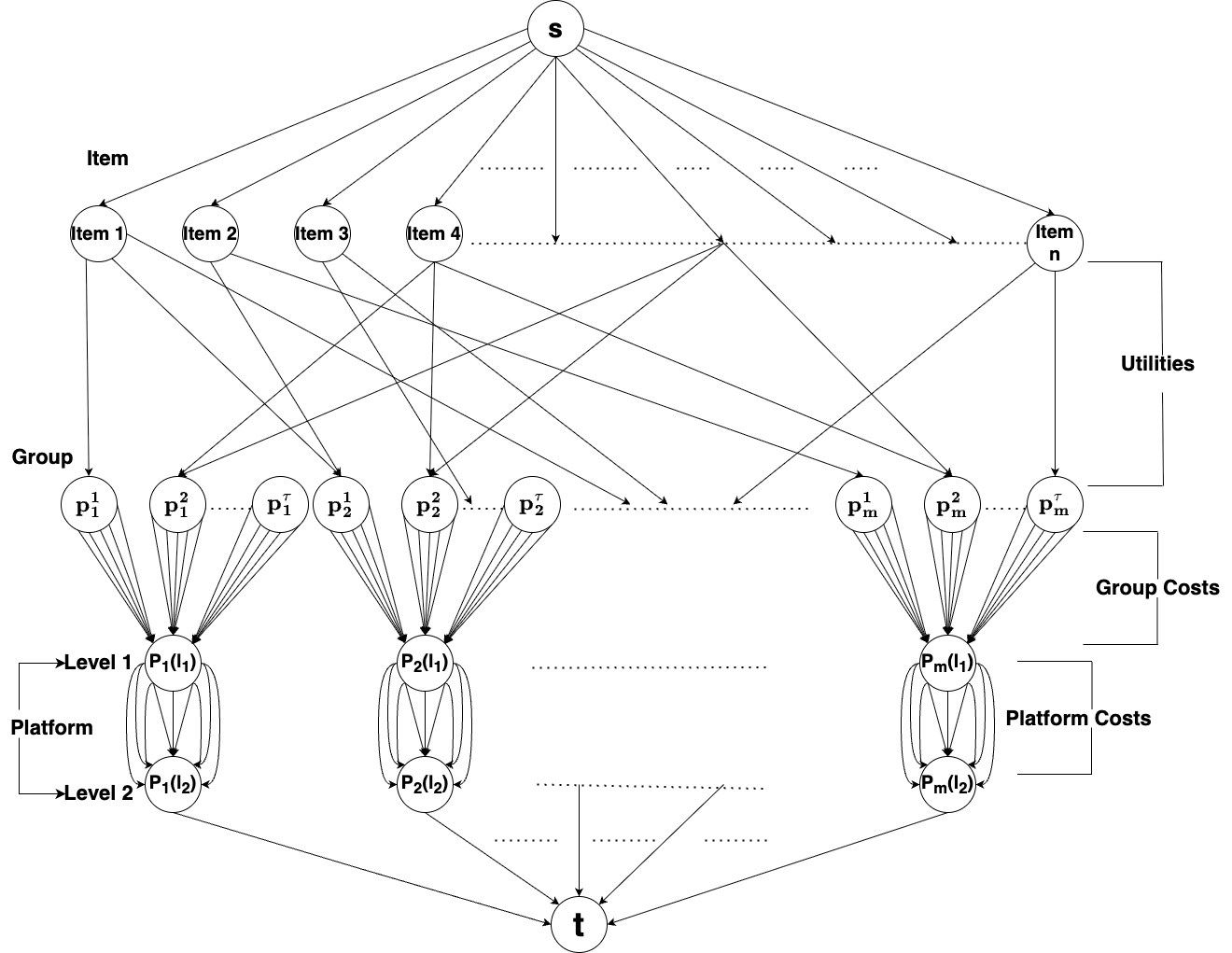}
\caption{Flow network when the input graph $\cGraph$ has $n$ items, $m$ platforms and $\tau$ groups with the following edges
$\cEdge{\text{Item }1}{\mathbf{p_1}}$, $\cEdge{\text{Item }1}{\mathbf{p_2}}$, $\cEdge{\text{Item }2}{\mathbf{p_2}}$, $\cEdge{\text{Item }4}{\mathbf{p_1}}$, $\cEdge{\text{Item }4}{\mathbf{p_2}}$, $\cEdge{\text{Item }4}{\mathbf{p_m}}$, $\cEdge{\text{Item }n}{\mathbf{p_m}}$. In $\cGraph$, Items $1$ and $2$ are in group $1$, Item $4$ is in group $2$ and Item $n$ is in group ${\tau}$}
\label{fig:flow_network}
\end{figure}

\noindent{\bf Edges:} Let the weight (cost) and capacity of any edge $\cEdge{v}{w}$ be denoted by $w_{vw}$ and $c_{vw}$ respectively. We add the following edges between the layers in the flow network. \\
    1. $\cEdge{s}{\cItem}$ with $w_{s\cItem} = 0$ and $c_{s\cItem} = 1$, $\forall \cItem \in \cItems$. These edges connect the source to all items, with a weight of $0$ and a capacity of $1$. \\
    2. $\forall \cItem \in \cItems, p \in N(i)$, we add the edges $\cEdge{\cItem}{\cPC{}{j}{}}$, $\forall j \in [\tau] : \cItem \in \cG_j$ with $c_{i\cPC{}{j}{}} = 1$ and $w_{i\cPC{}{j}{}} = 0$. Each of these edges with capacity $1$ and weight $0$ connects an item to the group-specific copy of every platform in its neighborhood, within the group layer. Additionally, these edges have a utility associated with them, $\cUtil_{\cItem\cPC{}{j}{}} = \cUtil_{\cItem p}$, $\forall \cEdge{\cItem}{\cPC{}{j}{}} \in \cFNE, \cEdge{\cItem}{p} \in E$. When the groups are disjoint, it is straightforward to verify that for every edge $(i,p) \in E$, there exists exactly one corresponding $(i,\cPC{}{j}{}) \in \cFNE$. We will later show that this one-to-one correspondence also holds in the case of laminar groups, as discussed in the network construction for that setting in \Cref{sec:laminar}. \\
    3. $\mathbf{L_1}$: This is the set of all edges between the group layer and the first level of the platform layer. $\forall p \in P$, $j \in g(p)$, we add an edge from $\cPC{}{j}{}$ to $\cPC{}{}{(l_1)}$ and make $\Delta_j(p)-1$ extra copies of this edge. Let $\cEdge{\cPC{}{j}{}}{\cPC{}{}{(l_1)}}^k$ denote the $k^{th}$ copy where $k \in [\Delta_j(p)]$. $\cEdge{\cPC{}{j}{}}{\cPC{}{}{(l_1)}}^k$ is an edge with capacity $1$ and weight $f_p^j(k) - f_p^j(k-1)$ $\forall k \in [\Delta_j(p)]$. These are parallel edges with the same capacity but different weights connecting all group-layer copies of a platform, $p \in P$, to $p(l_1)$ in the first level of the platform layer. \label{edge:L1} \\
    4. $\mathbf{L_2}$: This is the set of all edges between the first level and the second level of the platform layer. We add an edge $\cEdge{\cPC{}{}{(l_1)}}{\cPC{}{}{(l_2)}}$ $\forall p \in P$,  and make $\Delta(p)-1$ extra copies of this edge. Let $\cEdge{\cPC{}{}{(l_1)}}{\cPC{}{}{(l_2)}}^k$ denote the $k^{th}$ copy where $k \in [\Delta(p)]$. $\cEdge{\cPC{}{}{(l_1)}}{\cPC{}{}{(l_2)}}^k$ is an edge with capacity $1$ and weight $f_p(k) - f_p(k-1)$. These are parallel edges with the same capacity but different weights connecting the two copies of a platform in the first and second levels of the platform layer. \label{edge:L2} \\ 
    5. We add an edge $\cEdge{p}{t}$ with $w_{pt} = 0$ and $c_{pt} = \Delta(p)$, $\forall p \in P$. These edges connect each platform's copy in the second level of the platform layer to the target node $t$, with a weight of $0$ and a capacity of $\Delta(p)$.


Before presenting the Linear Programming formulation of the problem and our algorithm, 
we claim that there is a matching in $\cGraph$ where the minimum cost is $\hat{\mathcal{C}}_{\cI}$ to obtain utility at least $\hat{\ell}$, if and only if there is an integral flow on $\cFN$ with utility at least $\hat{\ell}$ and the minimum cost for such a flow is $\hat{\mathcal{C}}_{\cI}$. By construction, any such flow, when passed through steps \ref{step:begin_reduction}–\ref{step:end_reduction} of \Cref{alg:main}, yields a matching in $\cGraph$ with identical utility and cost. We formally state this below in \Cref{lem:mapping} and \Cref{claim:mapping}.

\begin{restatable}{lemma}{mapping}\label{lem:mapping}
    Let $\hat{\ell} \in \mathbb{R}_{\geq 0}$.
    There exists a matching in $\cGraph$ with utility at least $\hat{\ell}$ and minimum cost $\hat{\mathcal{C}}_{\cI}$ iff there exists a flow in the constructed network $\cFN$ 
    with utility at least $\hat{\ell}$ such that the minimum cost to achieve such a flow is $\hat{\mathcal{C}}_{\cI}$. 
\end{restatable}

\begin{restatable}{corollary}{corollaryMapping}\label{claim:mapping}
    Applying the reduction steps \ref{step:begin_reduction}–\ref{step:end_reduction} of \Cref{alg:main} to any flow on $\cFN$ with utility $\hat{\ell}$ and cost $\hat{\mathcal{C}}_{\cI}$ yields a matching in $\cGraph$ with utility $\hat{\ell}$ and cost $\hat{\mathcal{C}}_{\cI}$.
\end{restatable}
Before we provide a formal proof of \Cref{lem:mapping} and \Cref{claim:mapping} in \Cref{sec:Proofs}, we sketch the proof of \Cref{lem:mapping} here. 
\paragraph{Proof Sketch:} We establish a one-to-one correspondence between matchings in the bipartite graph, $\cGraph$ and flows in the constructed network $\cFN$, under the assumption that groups are disjoint. Each edge $\cEdge{i}{p} \in E$, where item $i$ belongs to group, $j$, maps uniquely to an edge from $i$ to the corresponding group-platform copy, $\cPC{}{j}{} \in \cFNV$, preserving utility. Hence, any matching in 
$\cGraph$ with utility at least $\hat{\ell}$
corresponds directly to a flow of the same utility in 
$\cFN$, and vice versa. We then show that the minimum cost of such a matching equals the minimum cost of the corresponding flow by accounting for the cost bearing edges in $\cFN$ to compute the total flow cost which exactly matches the cost expression of our objective in \Cref{eq:Objective}. \Cref{claim:mapping} directly follows from \Cref{lem:mapping}.

By \Cref{lem:mapping} and \Cref{claim:mapping}, it suffices to compute a flow on the constructed network, $\cFN$, that minimizes the total flow cost while ensuring that the flow utility is at least $\ell$ (\Cref{eq:Constraint}). The Integer Programming version of the Linear Program (LP) defined in  \Cref{LP_int} over $\cFN$ does exactly this. For any vertex $v \in \cFNV$, let $\delta^-(v)$ and $\delta^+(v)$ denote the sets of incoming and outgoing edges, respectively. It is easy to verify that the constraints \ref{LP_int-1}, \ref{LP_int-3}, \ref{LP_int-4} and \ref{LP_int-5} of \Cref{LP_int} describe the flow polytope corresponding to $\cFN$. We denote this flow polytope by $Q$ and refer the reader to standard references such as \cite{AhujaMO93} for background. 

\begin{lp}\label{LP_int}
 \begin{align}
     &\min \sum_{e \in \cFNE}w_ex_e  \label{LP_int-1} \\
     \text{such that} \quad &\sum_{e \in \delta^+(i)} x_e \leq 1 &\forall i \in I \label{LP_int-2}\\
     &\sum_{e \in \cFNE}\cUtil_ex_e \geq \ell \label{LP_int-3}\\
     &\sum_{e \in \delta^-(v)} x_e = \sum_{e \in \delta^+(v)} x_e &\forall v \in V \setminus \{s,t\} \label{LP_int-4}\\
     & x_e \geq 0 &\forall e \in \cFNE \label{LP_int-5}
 \end{align}
\end{lp}

Since \Cref{LP_int} is a linear programming relaxation, its optimal solution may be fractional. Consequently, the resulting flow may not correspond to a valid matching in the original graph. To address this, we apply a simple rounding procedure described in \Cref{alg:round} (see \Cref{sec:Proofs}) within \Cref{alg:main}, which converts the fractional flow into an integral one. This integral flow corresponds to a valid matching that satisfies the utility constraint while incurring only a bounded increase in cost. This is possible because we show that any optimal solution to \Cref{LP_int} contains at most two fractional paths. By rounding up the path with higher utility and rounding down the one with lower utility (if it exists), we obtain an integral solution that maintains feasibility and closely preserves the cost. Before we formally state these in \Cref{lem:frac_paths} and \Cref{lem:cost} and provide a combined proof sketch for both, we start with the following definition.

\begin{definition}\label{def:ext_point}
    An {\bf Extreme Point} of a polytope is a point that cannot be expressed as a convex combination of two distinct points in the polytope. Formally, let $Q \subseteq \mathbb{R}^n$ be a polytope. A point $x \in Q$ is an extreme point if there do not exist distinct $y,z \in Q, y \neq z$, and $\lambda \in (0,1)$, such that 
    $x = \lambda y + (1-\lambda)x$.
\end{definition}

In other words, extreme points are vertices of the polytope, and every bounded linear program attains its optimum at an extreme point of the feasible region.

\begin{algorithm}[t]
\caption{Min Cost $\ell$-util Flow($\cI = (\cGraph, \cG_1, \cG_2, \cdots, \cG_{\tau},F)$)}
\label{alg:main}
\nonl \textbf{Input} :  $\cI$ \\
\nonl \textbf{Output} : A matching, $\cM$, with utility at least $\ell$\\
$\cM = \emptyset$\\
Follow the steps in \Cref{subsec:Flow} to construct the flow network $\cFN$ \label{step:construct}\\
Solve \Cref{LP_int} on the network, $\cFN$, and store the result in $x^*$ \label{step:runLP} \\
\If{$x^*$ is fractional}{$FL =$ \hyperref[alg:round]{Rounding}($x^*$)} \label{step:round}
\Else{$FL=x^*$}
$FE = \{\cEdge{\cItem}{v}: \cItem \in \cItems \text{ and the edge } \cEdge{\cItem}{v} \text{ is part of the flow, } FL\}$ \\
\For{$\cEdge{\cItem}{v} \in FE$}{ \label{step:begin_reduction}
    Let $p$ be the platform such that $v$ is a group layer 
    copy of $p$.\\
    $\cM = \cM \cup \cEdge{\cItem}{p}$ \label{step:end_reduction}
}
Return $\cM$
\end{algorithm}

\begin{restatable}{lemma}{fracPaths}\label{lem:frac_paths}
    If an optimal solution of \Cref{LP_int} has fractional edges, it has at most two fractional paths.
\end{restatable}



\begin{restatable}{lemma}{cost}\label{lem:cost}
    \Cref{alg:main} returns a matching with total utility at least $\ell$ and cost at most $OPT + \max_p[\nabla\cCF{\sigma_p} + \max_j \nabla\cGCF{j}{\nu_p^j}] - \min_p[\nabla\cCF{\sigma_p} + \min_j \nabla\cGCF{j}{\nu_p^j}]$ when the groups are disjoint.
\end{restatable}

\paragraph{Proof Sketch:} We build on the framework of \cite{Gallo_Flow}, which characterizes the structure of extreme points of flow polytopes. We begin by recalling some definitions from \cite{Gallo_Flow}. A cycle in the flow network $\cFN$ is a sequence of nodes $(v_1, \ldots, v_r)$ such that for each $k$, either $(v_k, v_{k+1}) \in \cFNE$ or its reverse appears, with $v_1 = v_r$. $A(\gamma)$ denotes the set of arcs in a cycle, which is further partitioned into forward and reverse arcs based on direction. For a flow $x$, the set of \emph{floating arcs} is denoted by $A_1(x) = \{e \in \cFNE : 0 < x_e < \cUtil_e\}$. By Theorem 3.2 of \cite{Gallo_Flow}, if a cycle $\gamma$ is the only cycle in $A(\gamma) \cup A_1(x)$, then the flow can be perturbed along $\gamma$ to obtain an adjacent extreme point. Using properties of the flow network and the flow polytope, $Q$, we show that, any adjacent extreme points differ by at most one cycle.

Let $x^*$ be the optimal (fractional) solution to \Cref{LP_int}, then again using properties of the polytope $Q$ we show that $x^*$ is sandwiched between two adjacent integral extreme points, say $x_0, x_1$ , such that $x^* = \lambda x_0 + (1 - \lambda) x_1$. By the above, $x_0$ and $x_1$ differ on a single cycle, so $x^*$ has fractional flow only along that cycle, that is on at most two paths in $\cFN$. Therefore, we round up the path with higher utility and round down the other, preserving feasibility. The increase in cost is bounded by the difference between the maximum and minimum total cost across all such two-edge paths, leading to the claimed additive bound in \Cref{lem:cost}.

Since \Cref{lem:cost} establishes that \Cref{alg:main} returns a matching with total utility at least $\ell$ and cost at most $OPT + \max_p[\nabla\cCF{\sigma_p} + \max_j \nabla\cGCF{j}{\nu_p^j}] - \min_p[\nabla\cCF{\sigma_p} + \min_j \nabla\cGCF{j}{\nu_p^j}]$, the only thing required to complete the proof of \Cref{thm:disjoint} is to show is that \Cref{alg:main} runs in polynomial time. The total number of variables in \Cref{LP_int} is $|\cFNE|$ and the total number of constraints is $n + |\cFNV| + |\cFNE| - 1$.  Therefore, \Cref{LP_int} can be solved in time polynomial in the number of vertices and edges in the constructed network, $\cFN$. The runtime for both the rounding step (\Cref{alg:round}) and the loop from step \ref{step:begin_reduction} to \ref{step:end_reduction} in \Cref{alg:main} is $\mathcal{O}(|\cFNE|)$. Hence, if we can show that $\cFNE$ and $\cFNV$ are polynomial in the number of items, platforms, and edges in the input graph, $\cGraph$, we are done. We show this in \Cref{lem:runtime} stated below.

\begin{restatable}{lemma}{runtime}\label{lem:runtime}
    The flow network $\cFN$ has at most $n+2m+|E|+2$ nodes and exactly
    $n + m + 3|E|$ edges.
\end{restatable}

In \Cref{sec:Proofs}, we prove \Cref{lem:runtime} by accounting for the item, platform, group, and dummy layers and combining edges from the source, and the three internal layers in $\cFN$ (see \Cref{fig:flow_network}).

\subsection{Detailed Proof of Correctness}\label{sec:Proofs}
\begin{algorithm}[t]
\caption{Rounding Step($x^*$)}\label{alg:round}
\nonl \textbf{Input} :  An optimal solution of \Cref{LP_int} \\
\nonl \textbf{Output} : An integral flow on $\cFN$, $FL$ \\
\tcc{$\pi_1$ is the first fractional path and $\pi_2$ the second if any}
\tcc{$\pi_1$ is the path with higher utility w.l.o.g if $\pi_2$ exists} 
\For{$e \in \cFNE$}{
    \If{$e \in \pi_1$}{$x^*_e=1$} 
    \If{$e \in \pi_2$}{$x^*_e=0$}
    $FL_e = x^*_e$ 
}
Return $FL$
\end{algorithm}

In this section, we prove \Cref{lem:mapping}, \Cref{claim:mapping}, \Cref{lem:frac_paths}, \Cref{lem:cost}, and \Cref{lem:runtime} which together establish the proof of \Cref{thm:disjoint}. We begin with \Cref{lem:mapping} and \Cref{claim:mapping}, but first present an observation and a supporting lemma that will be used in their proofs.

\begin{restatable}{observation}{utilityandcost}\label{obs:path_cost}
In the flow network $\cFN$, constructed on some input bipartite graph, $\cGraph$, with disjoint groups, every $s$-$t$ path contains exactly one edge that contributes to the total utility, namely, the edge from the item layer to the group layer. Furthermore, each such path includes exactly two edges that contribute to the total cost: one edge from the group layer to the first level of the platform layer, and one edge from the first level to the second level of the platform layer.
\end{restatable}

\begin{proof}
By construction of $\cFN$, utilities are assigned to only the edges going from an item, say $i \in \cG_j$, to a platform copy $\cPC{}{j}{}$ in the group layer (see \Cref{fig:flow_network}). 
Again by construction of $\cFN$, costs are assigned only to edges that model assignment decisions influencing group-specific or overall platform load, specifically, edges from platform copies in the group layer to their corresponding platform replicas in the first level of the platform layer ($L_1$ edges defined in \ref{edge:L1}) and from these copies to their corresponding copies in the second level of the platform layer ($L_2$ edges defined in \ref{edge:L2}). All other edges, that is, from $s$ to items, from the item layer to the group layer, and from the platform layer to $t$, serve structural purposes and carry zero cost. Hence, any $s$-$t$ path traverses at most two cost-bearing edges.
\end{proof}

\begin{lemma}\label{lem:sequential}
    Consider an arbitrary platform $p$ and its copy $\cPC{}{j}{}$ under an arbitrary group $j \in g(p)$. In any optimal solution computed on $\cFN$ by \Cref{LP_int}, within the edge set $L_1$, if there is no flow through the edge $\cEdge{\cPC{}{j}{}}{\cPC{}{}{(l_1)}}^k$, then there will also be no flow through the edge $\cEdge{\cPC{}{j}{}}{\cPC{}{}{(l_1)}}^{k'}$, where $k < k' \leq \Delta_j(p)$, assuming that $\cGCF{j}{k} - \cGCF{j}{k-1} \neq \cGCF{j}{k'} - \cGCF{j}{k'-1}$. This relationship holds for both edge sets $L_1$ and $L_2$.
\end{lemma}
\begin{proof}
    Let us assume that there is no flow through $\cEdge{\cPC{}{j}{}}{\cPC{}{}{(l_1)}}^k$ but there is a flow through $\cEdge{\cPC{}{j}{}}{\cPC{}{}{(l_1)}}^{k'}$ where $k < k' \leq \Delta_j(p)$ in an optimal solution of \Cref{LP_int} on $\cFN$. This implies that there is an unused edge, $\cEdge{\cPC{}{j}{}}{\cPC{}{}{(l_1)}}^k$ with cost $\cGCF{j}{k} - \cGCF{j}{k-1}$. Since $\cGCF{j}{.}$ is a convex function, its slope is monotonically non-decreasing. Thus,
    $$\frac{\cGCF{j}{k} - \cGCF{j}{k-1}}{k - k + 1} \leq \frac{\cGCF{j}{k'} - \cGCF{j}{k'-1}}{k' - k' + 1}$$ 
    $$\implies \cGCF{j}{k} - \cGCF{j}{k-1} \leq \cGCF{j}{k'} - \cGCF{j}{k'-1}$$
   However, by the assumption in the lemma statement, $\cGCF{j}{k} - \cGCF{j}{k-1} < \cGCF{j}{k'} - \cGCF{j}{k'-1}$. Similarly, the edge $\cEdge{\cPC{}{}{(l_1)}}{\cPC{}{}{(l_2)}}^k$ is unused and we can show that $\cGCF{}{k} - \cGCF{}{k-1} < \cGCF{}{k'} - \cGCF{}{k'-1}$.
   Therefore, replacing the edges $\cEdge{\cPC{}{j}{}}{\cPC{}{}{(l_1)}}^{k'}$ and $\cEdge{\cPC{}{}{(l_1)}}{\cPC{}{}{(l_2)}}^{k'}$ with $\cEdge{\cPC{}{j}{}}{\cPC{}{}{(l_1)}}^k$ and  $\cEdge{\cPC{}{}{(l_1)}}{\cPC{}{}{(l_2)}}^k$ respectively, yields a strictly lower total cost without altering the overall utility (see \Cref{fig:flow_network}. This is because these edges do not contribute to the utility (\Cref{obs:path_cost})). This contradicts the assumption that we started with an optimal solution of \Cref{LP_int} in the first place, as a feasible solution with the same utility but strictly lower cost exists.
\end{proof}

    
We first restate \Cref{lem:mapping} then provide its formal proof.
\mapping*
\begin{proof}
    When the groups are disjoint, there is a one-to-one mapping between the edges from the item layer to the group layer in $\cFN$, and the original edges in the input graph $\cGraph$. Consider an arbitrary edge $\cEdge{i}{p} \in E$ where $i$ belongs to the group $\cG_j$ for some $j \in [\tau]$. In all of $\cFN$, there is exactly one edge, $\cEdge{i}{\cPC{}{j}{}} \in \cFNE$, from the item $i$ to exactly one copy of the platform $p$, $\cPC{}{j}{}$, in the group layer, with utility $\cUtil_{ip}$ by construction. Therefore, an edge $\cEdge{i}{\cPC{}{j}{}}$ contributes the exact same utility to a flow on $\cFN$ as the edge $\cEdge{i}{p}$ to a matching on $\cGraph$, hence iff there is a flow with utility at least $\hat{\ell}$, then there must be a matching with utility at least $\hat{\ell}$. Next, we will show that the minimum cost to achieve such a matching is the same as that to achieve such a flow.
    
    Let us assume that we have a flow with minimum cost, say $\hat{\mathcal{C}}_{\cI}$, such that the total utility is at least $\hat{\ell}$. Let $\Pi_{\hat{\ell}}$ denote the set of paths in this flow. Consider an arbitrary utility edge, $\cEdge{i}{\cPC{}{j}{}}$ in an arbitrary path $\pi \in \Pi_{\hat{\ell}}$. Since the cost of all edges from $s$ to the item layer, from items to the group layer, and from the platform layer to $t$ is $0$, let us look at the edges on the path, $\pi$, from $i$ to $t$ in $L_1 \cup L_2$. Let $\cEdge{i}{\cPC{}{j}{}}$ be the $\gamma^{th}$ incoming edge from items in $\cG_j$ to $\cPC{}{j}{}$ in this flow. Since all the outgoing edges from $\cPC{}{j}{}$ go only to $\cPC{}{}{(l_1)}$ in the first level of the platform layer, which is then connected to $\cPC{}{}{(l_2)}$, using the one-to-one mapping between $\cEdge{i}{\cPC{}{j}{}}$ and $\cEdge{i}{p}$, $\gamma$ represents the total number of items from group $\cG_j$ that are matched to the platform $p$ in the matching on bipartite graph $\cGraph$. Therefore, $\nu_p^j = \gamma$. Now, let us assume that there are non-consecutive edges in $L_1$ from $\cPC{}{j}{}$ to $\cPC{}{}{(l_1)}$ in the flow computed on $\cFN$. Without loss of generality, let these edges be the first $\gamma-1$ copies and the $z^{th}$ copy such that $\gamma < z \leq \Delta_j(p)$. 
    From \Cref{lem:sequential}, we know that this is possible only when the cost of the edge $\cEdge{\cPC{}{j}{}}{\cPC{}{}{(l_1)}}^{\gamma} == \cEdge{\cPC{}{j}{}}{\cPC{}{}{(l_1)}}^z$. Therefore, we can replace $\cEdge{\cPC{}{j}{}}{\cPC{}{}{(l_1)}}^z$ with $\cEdge{\cPC{}{j}{}}{\cPC{}{}{(l_1)}}^{\gamma}$ without any changes to the cost or utility. Hence, if there are $\gamma$ edges in the computed flow to $\cPC{}{}{(l_1)}$ from $\cPC{}{j}{}$, they must be to the first $\gamma$ copies. Assuming $\cGCF{j}{0} = 0$, the total cost from the edges between the group layer to the level $1$ of the platform layer for the platform, $p$ under group $I_j$ is 
    \begin{equation}\label{eq:group_cost}
        \cGCF{j}{1} + \cGCF{j}{2} - \cGCF{j}{1} + \cdots \cGCF{j}{\nu_p^j} - \cGCF{j}{\nu_p^j-1} = \cGCF{j}{\nu_p^j}
    \end{equation}  

    Using similar reasoning for the edges in $L_2$, if there are $\beta$ edges to the copies of $\cPC{}{}{(l_2)}$, these must be to the first $\beta$ copies and $\sigma_p = \beta$. Assuming $\cCF{0} = 0$, the total cost from the edges between the levels $1$ and $2$ of the platform layer for the platform, $p$ is
    \begin{equation}\label{eq:platform_cost}
        \cCF{1} + \cCF{2} - \cCF{1}+ \cdots \cCF{\sigma_p} - \cCF{\sigma_p - 1} = \cCF{\sigma_p}
    \end{equation}
    Thus, consolidating across all groups and platforms, from Equations \ref{eq:group_cost} and \ref{eq:platform_cost}, the total cost of the flow is
    $$\sum_{p \in P}\left(\sum_{j=1}^{\tau}\cGCF{j}{\nu_p^j} + \cCF{\sigma_p}\right)$$
    which matches the objective function $\cCost$ in \Cref{eq:Objective} which is the total cost of the matching on $\cGraph$. Therefore, $\hat{\mathcal{C}}_{\cI} = \min \sum_{p \in P}\left(\sum_{j=1}^{\tau}\cGCF{j}{\nu_p^j} + \cCF{\sigma_p}\right)$ such that utility is at least $\hat{\ell}$, hence the minimum cost of a matching on $\cGraph$ with utility at least $\hat{\ell}$ must also be $\hat{\mathcal{C}}_{\cI}$.
    
    Note that we do not need the assumptions $\cGCF{j}{0} = 0$ and $\cCF{0} = 0$ $\forall p \in P, j \in [\tau]$. If they are not $0$ as can be the case with a convex cost function with the minima not at the origin, the total cost of the flow would be 
    $$\sum_{p \in P}\left(\sum_{j=1}^{\tau}\left(\cGCF{j}{\nu_p^j} - \cGCF{j}{0}\right) + \cCF{\sigma_p} - \cCF{0}\right) = \cCost - \sum_{p \in P}\left(\sum_{j=1}^{\tau}\cGCF{j}{0} + \cCF{0}\right)$$
    Since, $\sum_{p \in P}\left(\sum_{j=1}^{\tau}\cGCF{j}{0} + \cCF{0}\right)$ is a constant, minimizing $\cCost - \sum_{p \in P}\left(\sum_{j=1}^{\tau}\cGCF{j}{0} + \cCF{0}\right)$ is the same as minimizing $\cCost$.
\end{proof}
We restate the \Cref{claim:mapping} before giving its proof.
\corollaryMapping*
\begin{proof}
    This follows directly from \Cref{lem:mapping}, since the matching computed through steps \ref{step:begin_reduction}–\ref{step:end_reduction} in \Cref{alg:main} is obtained by exploiting the one-to-one correspondence between the utility bearing edges from the item layer to the group layer in $\cFN$ and the original edges in the input graph $\cGraph$.
\end{proof}

Next we need the following observation for the proof of \Cref{lem:frac_paths} and \Cref{lem:cost} both of which are restated before their respective proofs.

\begin{observation}\label{obs:exact}
    An optimal solution to \Cref{LP_int} has utility exactly $\ell$.
\end{observation}
This follows because the costs are non-negative, and hence any solution with higher utility can be scaled down to another solution with utility exactly $\ell$ and smaller cost. The correctness of the rounding procedure (Algorithm~\ref{alg:round}) relies on Lemmas \ref{lem:frac_paths} and \ref{lem:cost} both of which are derived from the following result:

\begin{restatable}{lemma}{convexCombination}\label{obs:int_points}
    An optimal (fractional) vertex solution\footnote{An extreme point is a point in the polytope that cannot be expressed as a convex combination of two distinct points in the polytope. Equivalently, it is also a vertex of the polytope where as many linearly independent constraints as the polytope's dimension attain equality. See any reference on linear programming for details.} $x^*$ of \Cref{LP_int} is a convex combination of two integer solutions of the flow polytope on $\cFN$, denoted by $Q$. Moreover, one of the two solutions has both cost and utility larger than those of $x^*$ and the other one has both cost and utility smaller than those of $x^*$.
\end{restatable}
\begin{proof}
    It is known that the flow polytope is integral when all capacities are integers. Thus any extreme point of the polytope $Q$ is integral i.e. each vertex or extreme point of $Q$ corresponds to an integral flow (see e.g. \cite{AhujaMO93}). 
    By Observation~\ref{obs:exact}, an optimal solution of \Cref{LP_int} must lie on the hyperplane of \Cref{LP_int-3}. However, a vertex solution, like $x^*$, is formed by $|\cFNE|$ tight constraints. As one tight constraint is \Cref{LP_int-3}, the remaining $|\cFNE|-1$ tight constraints must be a part of the flow polytope $Q$, and they must define an edge of $Q$.
    Let the edge containing $x^*$ be between two vertices of the flow polytope denoted by $x^*_1$ and $x^*_2$. Hence $x^*$ is a convex combination of $x^*_1$, $x^*_2$, i.e. $\exists \lambda\in(0,1)$ such that $x^*=\lambda x^*_1+(1-\lambda)x^*_2$. Consequently, the cost and utility of $x^*$ must be a convex combination of the costs and utilities of $x^*_1$ and $x^*_2$ respectively.

    Since $x^*$ lies on the hyperplane defined by \Cref{LP_int-3}, and is a convex combination of $x^*_1$ and $x^*_2$, $x^*_1$ and $x^*_2$ must lie on the opposite side of the hyperplane of \Cref{LP_int-3}. Hence one, say $x^*_1$, must have utility larger than $\ell$ and $x^*_2$ must have utility smaller than $\ell$. If $x^*_2$ has a larger cost than $x^*_1$, then by increasing $\lambda$ in the convex combination, one gets a larger utility at a smaller cost, contradicting the optimality of $x^*$. Hence the second part of the lemma statement follows.
\end{proof}

\fracPaths*
\begin{proof}
    We start with a few definitions that are derived from \cite{Gallo_Flow}. We call a sequence $(v_1,v_2, \cdots, v_r)$ of nodes of $\cFN$, a cycle if for $k = 1,2, \cdots, r-1$, either $(v_{k}, v_{k+1}) \in \cFNE$ or $(v_{k+1}, v_{k}) \in \cFNE$ and $v_1 = v_r$.
    Given a cycle, say $\gamma$, we denote the arcs in $\gamma$ by $A(\gamma)$. Similar to \cite{Gallo_Flow}, we partition the arcs of a cycle into the following:
    \begin{align*}
        & \textbf{Forward Arcs: }&A^+(\gamma) = \{a \in \cFNE: a = (v_k, v_{k+1}) \text{ for some } 1 \leq k \leq r-1\} \\
        & \textbf{Reverse Arcs: }&A^-(\gamma) = \{a \in \cFNE: a = (v_{k+1}, v_{k}) \text{ for some } 1 \leq k \leq r-1\}
    \end{align*}
    The authors in \cite{Gallo_Flow} also define something known as floating arcs for any flow solution, say $x$, formally defined below:
    \begin{equation}\label{eq:floating_arcs}
        A_1(x) = \{e \in A: 0 < x_e < \cUtil_e\}
    \end{equation}
    
    Let $x$ be an extreme point of the flow polytope. From Theorem 3.2 of \cite{Gallo_Flow}, we know that if there is a cycle, say $\gamma$ such that $A(\gamma)\cup A_1(x)$ has $\gamma$ as its only cycle, then $y=x + \mu(\gamma)\cdot\lambda$ is an extreme flow adjacent to $x$ where $\lambda > 0$, and $\mu(\gamma)$ is a $|\cFNE|$ vector with  $\mu(e) = 1$ if an edge $e$ in the cycle, $\gamma$ is a forward arc in $\cFN$ and $\mu(e) = -1$ if $e \in \gamma$ is a reverse arc in $\cFN$. Therefore, $y$ differs from $x$ by, at most, a cycle. Note that the capacity is $1$ for all the edges in $\cFN$ (an edges of capacity $\Delta(p)$ from $p$ to $t$ can be considered as $\Delta(p)$ parallel edges with capacity $1$) and any extreme point in $Q$ is integral; therefore, any extreme flow has edges with a value  $0$ or $1$, which means there are no floating arcs (\Cref{eq:floating_arcs}) in any extreme points of the polytope $Q$. Therefore, for any cycle $\gamma$, $A(\gamma)\cup A_1(x)$ has $\gamma$ as its only cycle, trivially.

    Let the optimal solution of \Cref{LP_int} be $x^*$, and let the two integer solutions of $Q$ that sandwich $x^*$ be $x_0$ and $x_1$, respectively. Since, $x^*$ can be written as a convex combination of $x_0$ and $x_1$ by \Cref{obs:int_points}, and $x_0$ and $x_1$ differ by only a cycle, $x^*$ can have fractional values only in one cycle which is two paths in $\cFN$.
\end{proof}

\cost*

\begin{proof}
    Let the two fractional paths in an optimal solution of \Cref{LP_int} (\Cref{lem:frac_paths}), say $x^*$, be denoted by $\pi_1$ and $\pi_2$, and let $\pi_2$ be the path with higher utility compared to $\pi_1$ w.l.o.g. Let the costs on the two edges (\Cref{obs:path_cost}) on the path $\pi_1$ be denoted by $w_{\pi_1^1}$ and $w_{\pi_1^2}$, and the costs on $\pi_2$ be denoted by $w_{\pi_2^1}$ and $w_{\pi_2^2}$. Let the fractional values assigned to $\pi_1$ and $\pi_2$ be $x^*_{\pi_1}$ and $x^*_{\pi_2}$ respectively. By \Cref{obs:int_points}, $x^*$ can be written as a convex combination of two adjacent integral flows, and by \Cref{lem:frac_paths}, these integral points only differ by the cycle formed by combining $\pi_1$ and $\pi_2$. Therefore, $x^*_{\pi_2} = 1- x^*_{\pi_1}$. Since the rounding step (\Cref{alg:round}) used in \Cref{alg:main} rounds up $x^*_{\pi_2}$ and rounds down $x^*_{\pi_1}$, the change in cost is going to be 
    $$[(w_{\pi_2^1}+w_{\pi_2^2}) - (w_{\pi_1^1}+w_{\pi_1^2})]\cdot x^*_{\pi_1}$$
    
    Without loss of generality, let $x_0$ be the integer solution with both utility and cost lower than $x^*$. Since $x_1$ is on the other side of the intersection point of the constraint \ref{LP_int-3} and the flow polytope $Q$, its utility and cost must be higher than both $x_0$ and $x^*$. Therefore, $x_1$ satisfies constraint \ref{eq:Constraint}. The rounding step in \Cref{alg:main} computes $x_1$ from $x^*$. Therefore, the flow computed in step \ref{step:round} of \Cref{alg:main} returns an integer flow that satisfies the utility constraint. Hence, by \Cref{claim:mapping},
    \Cref{alg:main} returns a matching with total utility at least $\ell$ and cost at most $\sum_{e \in \cFNE}w_ex^*_e + [(w_{\pi_2^1}+w_{\pi_2^2}) - (w_{\pi_1^1}+w_{\pi_1^2})]\cdot x^*_{\pi_1}$. Let the set of all possible paths in $\cFN$ be denoted by $\Pi$. By constraint \ref{LP_int-2}, $x^*_{\pi_1} \leq 1$, then,
    $$[(w_{\pi_2^1}+w_{\pi_2^2}) - (w_{\pi_1^1}+w_{\pi_1^2})]\cdot x^*_{\pi_1} \leq \max_{\pi \in \Pi}(w_{\pi^1}+w_{\pi^2}) - \min_{\pi \in \Pi}(w_{\pi^1}+w_{\pi^2})$$
    Since \Cref{LP_int} is an LP relaxation of the integer programming version of our problem, $\sum_{e \in \cFNE}w_ex^*_e \leq OPT$. By construction of the network, $\cFN$, $\forall \pi \in \Pi$
    $$\min_j \nabla\cGCF{j}{\nu_p^j} \leq w_{\pi^1} \leq \max_j \nabla\cGCF{j}{\nu_p^j} $$
    and
    $$ \min_j \nabla\cCF{\sigma_p}{} \leq w_{\pi^2} \leq \max_j \nabla\cCF{\sigma_p}{}$$ Therefore, \Cref{alg:main} returns a matching with total utility at least $\ell$ and cost at most 
    $$OPT + \max_p[\nabla\cCF{\sigma_p} + \max_j \nabla\cGCF{j}{\nu_p^j}] - \min_p[\nabla\cCF{\sigma_p} + \min_j \nabla\cGCF{j}{\nu_p^j}]$$
\end{proof} 

Finally, we restate and prove \Cref{lem:runtime} completing all the essential components required for the proof of \Cref{thm:disjoint}.

\runtime*

\begin{proof}
    The item layer has $n$ nodes, and the platform layer has $2$ copies of each platform, which is $2m$ nodes. Accounting for the dummy nodes, $s$ and $t$, we have $n+2m+2$ nodes that do not include the group layer. In the group layer, each platform is copied at most $\Delta(p)$ times since the number of groups from which items are in $N(p)$ cannot exceed $|N(p)|$, which is the same as $\Delta(p)$. Therefore, in the group layer, we have at most $\sum_{p \in P}\Delta(p) = |E|$ nodes. Therefore, the flow network, $\Gamma$ has at most $n+2m+|E|+2$ nodes.
    
    There are $n$ edges, one for each item from $s$ and $m$ edges, one from each platform in the second level of the platform layer to $t$. Each item has an edge to exactly one copy of every platform in its neighborhood; therefore, the total number of edges between the item layer and the group layer is $\sum_{i \in \cItems}\Delta(i) = |E|$. The edge set $L_1$ has $\Delta_j(p)$ edges from $\cPC{}{j}{}$ to $\cPC{}{}{(l_1)}$, $\forall p \in P, j \in g(p)$. Since the groups are disjoint, $\sum_{j \in g(p)}\Delta_j(p) = \Delta(p)$. Thus, the total number of edges in $L_1$ is $\sum_{p \in P}\Delta(p) = |E|$. The edge set $L_2$ has $\Delta_p$ edges from $p(l_1)$ to $p(l_2)$, $\forall p \in P$, therefore the total number of edges in $L_2$ is $\sum_{p \in P}\Delta(p) = |E|$. So, the total number of edges in $\cFN$ is $n + m + 3|E|$.
\end{proof}

      

\begin{proof}[Proof of \Cref{thm:disjoint}]
    We have already established in \Cref{lem:cost} that \Cref{alg:main} returns a matching with total utility at least $\ell$ and cost at most $OPT + \max_p[\nabla\cCF{\sigma_p} + \max_j \nabla\cGCF{j}{\nu_p^j}] - \min_p[\nabla\cCF{\sigma_p} + \min_j \nabla\cGCF{j}{\nu_p^j}]$. The total number of variables in \Cref{LP_int} is $|\cFNE|$ and the total number of constraints is $n + |\cFNV| + |\cFNE| - 1$. Therefore, by \Cref{lem:runtime}, the total number of variables in \Cref{LP_int} is $n + m + 3|E|$ and the total number of constraints is at most $3n+3m+4|E|+1$. Therefore, \Cref{LP_int} can be solved in time polynomial in the number of vertices and edges in the input graph, $\cGraph$. The runtime for both the rounding step (\Cref{alg:round}) and the loop from step \ref{step:begin_reduction} to \ref{step:end_reduction} in \Cref{alg:main} is $\mathcal{O}(|\cFNE|)$. Therefore, \Cref{alg:main} is polynomial in the number of items, platforms, and edges in the input graph, $\cGraph$.
\end{proof}

\section{Cost-Approximate Matching for Laminar Groups}\label{sec:laminar}
We prove \Cref{thm:laminar} in this section. A family of sets, say $S$, is laminar if, for every pair of sets $X, Y \in S$, one of the following holds: $X \subseteq Y$ or $Y \subseteq X$ or $X \cap Y = \phi$. Please refer to \Cref{fig:laminar_groups} for an example of a set of groups that follows a laminar structure. This example will serve as a running illustration to explain the changes required in the flow network construction when the groups are laminar. Note that, under a laminar structure, all the groups of $\cGraph$ would form a forest. We modify the flow network, $\cFN$, when the groups are laminar and not disjoint, and denote the modified flow network by $\Gamma_L$. The following modifications to $\Gamma$ are required when the groups are laminar.

\begin{figure}[t]
\centering
    \begin{tikzpicture}[
      every node/.style = {rectangle, draw},
      level distance = 1.5cm,
      level 1/.style = {sibling distance=3cm},
      level 2/.style = {sibling distance=1.5cm},
      level 3/.style = {sibling distance=1.5cm},
      edge from parent/.style = {draw, -latex}
    ]
    
    \node (group1) {$\cG_1$}
      child {
        node {$\cG_2$} 
          child { node {$\cItem_1$} }
          child { node {$\cItem_2$} }
      }
      child {
        node {$\cG_3$} 
          child { node {$\cItem_3$} }
          child {
            node {$\cG_4$} 
              child { node {$\cItem_5$} }
              child { node {$\cItem_6$} }
          }
          child { node {$\cItem_4$} }
      };
    
    \end{tikzpicture}
    
    \caption{Laminar structure of the set of groups $\{\cG_1, \cG_2,\cG_3,\cG_4\}$, where, $\cG_1 = \{\cItem_1, \cItem_2, \cItem_3,\cItem_4,\cItem_5,\cItem_6 \}, \cG_2 = \{\cItem_1, \cItem_2 \}, \cG_3 = \{\cItem_3,\cItem_4,\cItem_5,\cItem_6 \}$, and $\cG_4 = \{\cItem_5, \cItem_6 \}$.}
    \label{fig:laminar_groups}
\end{figure}
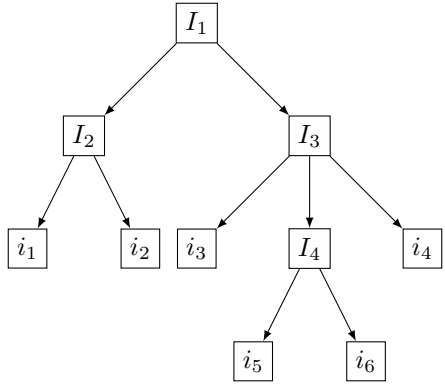

\begin{enumerate}[label=\arabic*., left=0pt, labelsep=1em, itemsep=0pt, wide=0pt, align=parleft]
    \item Let the maximum depth of any tree in the forest formed by the laminar structure of all the groups in $\cGraph$ be denoted by $d$. The flow network $\cFN_L$, which is a multi-layered, directed bipartite multi-graph, would now have $d$ levels in the group layer instead of just $1$. The Item and Platform layers remain unchanged. Copies of platforms under groups that do not have any other groups as their subsets are positioned in the first level of the group layer, directly adjacent to the Item layer. If a group, say $j \in [\tau]$, contains $t$ other groups, copies of platforms under $j$ will be placed in the $(t+1)^{th}$ level.
    For instance, in \Cref{fig:laminar_groups}, the groups $\cG_4$ and $\cG_2$ would be placed in the first level, $\cG_3$ in the second level, and $\cG_1$ in the third level (refer \Cref{fig:laminar_flow}). This structure ensures that groups in each level of the group layer are disjoint. \label{lam_depth}
    \item \label{group_layer_construction} Any item, say $\cItem$, in the Item layer will have edges connecting it to the copies of platforms in its neighborhood only within the smallest group to which it belongs. For instance, in \Cref{fig:laminar_groups}, suppose $(\cItem_5,p) \in E$ for some platform $p \in P$. The item $\cItem_5$ belongs to groups $\cG_4, \cG_3,$ and $\cG_1$, however, it will have edges to the copy of $p$ only under $\cG_4$ in the group layer, that is, to $\cPC{}{4}{}$ since $\cG_4$ does not contain any smaller groups within itself as shown in \Cref{fig:laminar_flow}.
    \item Let the forest formed by the groups in the input graph, $\cGraph$, be denoted by $LF$. Consider the copy of an arbitrary platform, $p \in P$, under an arbitrary group, say $j \in g(p)$. The node $\cPC{}{j}{}$ receives incoming edges from all its child nodes in $LF$ and has an equal number of outgoing edges directed towards its parent node in $LF$. For instance, see Figures \ref{fig:laminar_groups} and \ref{fig:laminar_flow}. Consequently, each node in the group layer has the following two types of edges:
    \begin{enumerate}[label=\arabic*., left=0pt, labelsep=1em, itemsep=0pt, wide=0pt, align=parleft]
        \item {\bf Item Edges:} These edges originate from the item layer. Each edge has a cost of $0$ and a capacity of $1$.
        \item {\bf Subset Edges:} These edges represent connections from a platform copy in one group to the corresponding platform copy in its parent group within $LF$. Specifically, consider an arbitrary group $j \in [\tau]$, and an arbitrary child group of $j$ in $LF$, say
        $j' \in [\tau]$. For instance, in \Cref{fig:laminar_groups}, group $3$ is a parent of group $4$. For any platform, $p \in P$, we introduce the edges $\cEdge{\cPC{}{j'}{}}{\cPC{}{j}{}}$ and make $\Delta_{j'}(p)$-1 additional copies of it. Each edge has a capacity of $1$, and the cost of the $k^{th}$ copy denoted by $\cEdge{\cPC{}{j'}{}}{\cPC{}{j}{}}^k$ is given by $\cGCF{j'}{k} - \cGCF{j'}{k-1}$.
    \end{enumerate}
\end{enumerate}

\begin{figure}[t]
\centering
    \includegraphics[scale=0.25]{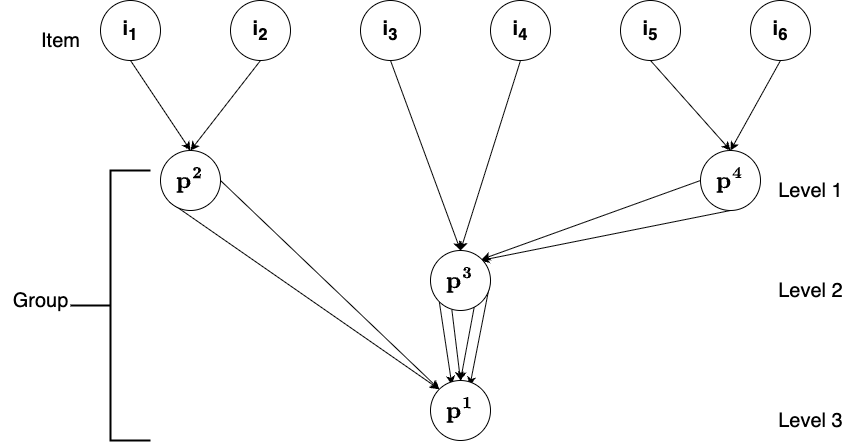}
    \caption{Item and group layers for the platform $p$ in the flow network, $\cFN_L$, constructed from the input graph, $\cGraph$, with $N(p) = \{i_1,i_2,i_3,i_4,i_5,i_6\}$, and the groups following the laminar structure shown in \Cref{fig:laminar_groups}.}\label{fig:laminar_flow}
\end{figure}

\subsection{Detailed Proof of Correctness}
Here, we prove \Cref{thm:laminar}. First we list certain properties of the modified flow network. In the modified flow network structure, despite the addition of new edges between different levels in the group layer, the following properties remain valid:
\begin{enumerate}[label=\arabic*., left=0pt, labelsep=1em, itemsep=0pt, wide=0pt, align=parleft]
    \item {\bf Platform Copy Edge Balancing:} Any platform copy in, say, the $j^{th}$ group still receives a total of $\Delta_j(p)$ incoming edges. These edges may come either from item edges originating at the item layer or from subset edges coming from its child nodes in the forest $LF$. Similarly, $\cPC{}{j}{}$ maintains $\Delta_j(p)$ outgoing edges, which are directed either to its parent node in $LF$ or to the first level of the platform layer.
    \item {\bf Item Connectivity to Final Level:} Every item in the item layer is connected to the last $(d^{th})$ level of the group layer either directly through item edges or indirectly through a combination of item edges and subset edges that traverse the laminar structure (see \Cref{fig:laminar_flow} for reference).
    \item {\bf Leveled Disjoint Group Structure:} The disjoint property of groups is preserved within each level of the Group layer (see \Cref{fig:laminar_flow}).
    
\end{enumerate}
The transition to laminar group structures involves rearranging the Group layer into multiple levels while maintaining disjointness at each level. Additionally, introducing subset edges acts as connectors between the levels without altering the essential properties of the flow network outside the group layer. 

As a result, {\bf \Cref{lem:sequential}} continues to hold: The sequential flow property holds, as the cost and capacity of edges between the Group and Platform layers are unchanged, and the additional subset edges only serve as connectors without disrupting the sequential order of flow through the levels.

The correctness of the algorithm for laminar group structures now hinges on verifying \Cref{lem:mapping}, and on establishing an upper bound on the cost violation and the size of the modified network $\cFN_L$, analogous to \Cref{lem:cost}) and \Cref{lem:runtime} in the disjoint case. We begin by proving \Cref{lem:mapping} under the modified flow network structure.

\begin{proof}[Proof of \Cref{lem:mapping} for Laminar Groups]
    By construction, the one-to-one correspondence between any edge, say $\cEdge{i}{\cPC{}{j}{}}$ from the item to the group layer and the edge $\cEdge{i}{p}$ is still preserved. This holds because any item in the item layer will have edges connecting it to the copies of platforms in its neighborhood only within the smallest group to which it belongs (refer \Cref{fig:laminar_flow}). Therefore, there is a matching with utility at least $\hat{\ell}$ iff there is a flow with utility at least $\hat{\ell}$.
    
    Consider an edge $\cEdge{i}{p}$ such that $i$ belongs to $\beta$ groups. When groups overlap, the critical distinction is how costs are computed compared to the disjoint case. The inclusion of the edge $\cEdge{i}{p}$ increments the count of items matched to $p$ from all its $\beta$ groups. Therefore, the total cost associated with the edge $\cEdge{i}{p}$ must account for the convex cost functions associated with all these $\beta$ groups. Let $\gamma_p^j$ be the number of items matched to the platform $p$, from group, $j$ where $j$ is one of the $\beta$ groups, $i$ belongs to, before the edge, $\cEdge{i}{p}$ was included in the matching. Let $\gamma_p$ be the number of items matched to $p$ before $\cEdge{i}{p}$ was included in the matching. Then $\cEdge{i}{p}$ adds the following total cost:
    \begin{equation}\label{eq:laminar_cost}
        \cCF{\gamma_p+1} - \cCF{\gamma} + \sum_{j: i \in \cG_j} \cGCF{j}{\gamma_p^j+1} - \cGCF{j}{\gamma_p^j}
    \end{equation}
    The cost due to platform associated convex cost functions are not impacted by overlapping groups.
    The edge $\cEdge{i}{p}$ in a matching on $\cGraph$ corresponds to the edge $\cEdge{i}{\cPC{}{j}{}}$ on the corresponding flow, say $F$ on $\cFN_L$ where $\cG_j$ is the smallest group $i$ belongs to. A $s-t$ path, $\pi$, passing through the item $i$ in the item layer, traverses the group layer via the copies of the platform $p$ associated with all the groups to which $i$ belongs, and continues through both $\cPC{}{}{(l_1)}$ and $\cPC{}{}{(l_2)}$ of the platform layer. By the construction of $\cFN_L$, this path goes through the copy of the platform in all the groups to which $i$ belongs. Since the groups form a laminar structure, the edge in $\pi$ that goes from a child node in $LF$, say group $j$ to its parent node in $LF$ say, group $j'$ would have a cost of $\cGCF{j}{k} - \cGCF{j}{k-1}$ if it was the $k^{th}$ copy of the edge. Similarly, for all the other groups to which $i$ belongs and the edges in $L_2$ between $\cPC{}{}{(l_1)}$ and $\cPC{}{}{(l_2)}$. Therefore, by \Cref{lem:sequential}, the total cost added by including the path, $\pi$, in the flow is 
    $$\cCF{\gamma_p+1} - \cCF{\gamma} + \sum_{j: i \in \cG_j} \cGCF{j}{\gamma_p^j+1} - \cGCF{j}{\gamma_p^j}$$
    which is the same as \Cref{eq:laminar_cost}.
\end{proof}

Having established \Cref{lem:mapping} for laminar groups, it follows directly that \Cref{claim:mapping} also holds in this setting.

We now state and prove a lemma analogous to \Cref{lem:cost} from the disjoint group case.

\begin{lemma}\label{lem:lam_cost}
    \Cref{alg:main} returns a matching with total utility at least $\ell$ and cost at most $OPT + \max_p[\nabla\cCF{\sigma_p} + \sum_{r=1}^{d}\max_{j_r} \nabla\cGCF{j_r}{\nu_p^{j_r}}] - \min_p[\nabla\cCF{\sigma_p} + \min_j \nabla\cGCF{j}{\nu_p^j}]$ when the groups are laminar.
\end{lemma}

\begin{proof}
    Let the two fractional paths in an optimal solution of \Cref{LP_int} (\Cref{lem:frac_paths}), say $x^*$, be denoted by $\pi_1$ and $\pi_2$, and let $\pi_2$ be the path with higher utility compared to $\pi_1$ w.l.o.g. Note that now the costs are not just on two edges of any $s-t$ path since there are weighted edges between the different levels within the group layer itself. In the worst case $\pi_2$ would have $d$ levels in the group layer of $\cFN_L$ (see \Cref{lam_depth}) and $\pi_1$ would have $1$ level.
    Let the costs on the edges going from level $i$ to $i+1$ on the path $\pi_2$ be denoted by $w_{\pi_2^i}$ and the costs on the two edges \Cref{obs:path_cost} on the path $\pi$ be denoted by $w_{\pi_1^1}$ and $w_{\pi_1^1}$. For $i = d, d+1$, $w_{\pi_2^i}$ denotes the cost associated with edges from the group layer to the first level of the platform layer (when $i = d$) and from the first level to the second level of the platform layer (when $i = d+1$), respectively.
    Let the fractional values assigned to $\pi_1$ and $\pi_2$ be $x^*_{\pi_1}$ and $x^*_{\pi_2}$ respectively. By \Cref{obs:int_points}, $x^*$ can be written as a convex combination of two adjacent integral flows, and by \Cref{lem:frac_paths}, these integral points only differ by the cycle formed by combining $\pi_1$ and $\pi_2$. Therefore, $x^*_{\pi_2} = 1- x^*_{\pi_1}$. Since the rounding step (\Cref{alg:round}) used in \Cref{alg:main} rounds up $x^*_{\pi_2}$ and rounds down $x^*_{\pi_1}$, the change in cost is going to be 
    $$\left[\sum_{i=1}^{d+1}w_{\pi_2^i} - (w_{\pi_1^1}+w_{\pi_1^2})\right]\cdot x^*_{\pi_1}$$
    
    Without loss of generality, let $x_0$ be the integer solution with both utility and cost lower than $x^*$. Since $x_1$ is on the other side of the intersection point of the constraint \ref{LP_int-3} and the flow polytope $Q$, its utility and cost must be higher than both $x_0$ and $x^*$. Therefore, $x_1$ satisfies constraint \ref{eq:Constraint}. The rounding step in \Cref{alg:main} computes $x_1$ from $x^*$. Therefore, the flow computed in step \ref{step:round} of \Cref{alg:main} returns an integer flow that satisfies the utility constraint. Hence, by \Cref{claim:mapping},
    \Cref{alg:main} returns a matching with total utility at least $\ell$ and cost at most $\sum_{e \in E_{\cFN_L}}w_ex^*_e + \left[\sum_{i=1}^{d+1}w_{\pi_2^i} - (w_{\pi_1^1}+w_{\pi_1^2})\right]\cdot x^*_{\pi_1}$. Let the set of all possible paths in $\cFN_L$ be denoted by $\Pi$. By constraint \ref{LP_int-2}, $x^*_{\pi_1} \leq 1$, then,
    $$\left[\sum_{i=1}^{d+1}w_{\pi_2^i} - (w_{\pi_1^1}+w_{\pi_1^2})\right]\cdot x^*_{\pi_1} \leq \max_{\pi \in \Pi}\left(\sum_{i=1}^{d+1}w_{\pi_2^i}\right) - \min_{\pi \in \Pi}(w_{\pi^1}+w_{\pi^2})$$
    Since \Cref{LP_int} is an LP relaxation of the integer programming version of our problem, $\sum_{e \in E_{\cFN_L}}w_ex^*_e \leq OPT$. By construction of the network,
    $$\min_j \nabla\cGCF{j}{\nu_p^j} \leq w_{\pi_1^1} \text{ and } \min_j \nabla\cCF{\sigma_p}{} \leq w_{\pi_1^2}$$
     
    Similarly, let $j_r$ denote a group at the $r^{th}$ level, then
    $$ \sum_{i=1}^{d} w_{\pi_2^i} \leq \sum_{r=1}^{d}\max_{j_r} \nabla\cGCF{j_r}{\nu_p^{j_r}} \text{ and }  w_{\pi_2^{d+1}} \leq \max_j \nabla\cCF{\sigma_p}{}$$ 
    Therefore, \Cref{alg:main} returns a matching with total utility at least $\ell$ and cost at most 
    $$OPT + \max_p\big[\nabla\cCF{\sigma_p} + \sum_{r=1}^{d}\max_{j_r} \nabla\cGCF{j_r}{\nu_p^{j_r}}\big] - \min_p[\nabla\cCF{\sigma_p} + \min_j \nabla\cGCF{j}{\nu_p^j}]$$
\end{proof} 

The final step in proving \Cref{thm:laminar} is to compute the total number of nodes and edges in the modified network $\cFN_L$, which allows us to establish that \Cref{alg:main} runs in polynomial time even when the groups follow a laminar structure.

\begin{lemma}\label{lem:lam_runtime}
    The flow network $\cFN_L$ has at most $n+2m+|E|+2$ nodes and at most
    $n + m + (d+2)|E|$ edges, where $d$ is the maximum number of groups an item can belong to.
\end{lemma}
\begin{proof}
    The number of nodes remains the same as in the disjoint group setting, since no additional copies of nodes in the input graph, $\cGraph$, are created when constructing $\cFN_L$. The only difference in edges is within the group layer, where items still have only one edge to the group layer but there are edges within the group layer between the different levels. Let the maximum number of groups an item can belong to be $d$ that is the maximum depth of any tree in the laminar forest (\Cref{fig:laminar_groups}) is at most $d$.
    So, an item introduces one edge to the group layer and $d - 1$ more edges within the group layer
    that pass through children groups in the laminar tree. Therefore, the maximum number of edges introduced by the items is $\sum_{i \in \cItems}d\cdot\Delta(i) = d\cdot|E|$.
    The edge sets $L_1$ and $L_2$ are untouched in $\cFN_L$, therefore the total number of edges in $L_1 \cup L_2$ stays $2|E|$ as shown in the proof of \Cref{lem:runtime}. So, the total number of edges in $\cFN$ is at most $n + m + (d+2)|E|$.
\end{proof}

\begin{proof}[Proof of \Cref{thm:laminar}]
From \Cref{lem:mapping} and \Cref{lem:lam_cost}, \Cref{alg:main} returns a matching with utility at least $\ell$ and cost at most $OPT + \max_p\big[\nabla\cCF{\sigma_p} + \sum_{r=1}^{d}\max_{j_r} \nabla\cGCF{j_r}{\nu_p^{j_r}}\big] - \min_p[\nabla\cCF{\sigma_p} + \min_j \nabla\cGCF{j}{\nu_p^j}]$. The total number of variables in \Cref{LP_int} is $|\cFNE|$ and the total number of constraints is $n + |\cFNV| + |\cFNE| - 1$. Therefore, by \Cref{lem:lam_runtime}, the total number of variables in \Cref{LP_int} is $n + m + 3|E|$ and the total number of constraints is at most $3n+3m+4|E|+1$. Therefore, \Cref{LP_int} can be solved in time polynomial in the number of vertices and edges in the input graph, $\cGraph$. The runtime for both the rounding step (\Cref{alg:round}) and the loop from step \ref{step:begin_reduction} to \ref{step:end_reduction} in \Cref{alg:main} is $\mathcal{O}(|\cFNE|)$. Therefore, \Cref{alg:main} is polynomial in the number of items, platforms, and edges in the input graph, $\cGraph$.
\end{proof}

\section{Hardness for non-laminar groups}\label{sec:hardness}

Our hardness reduction is inspired by a similar reduction from \cite{GroupFairMatching}, where they  consider a problem that is similar to $\cProblem$, but instead of convex cost functions, there is a strict quota constraint, and their problem is called Classified Maximum Matching (CMM). 

\begin{proof}[Proof of Theorem~\ref{thm:hardness}]
Our reduction from maximum independent set problem is described below.
\begin{itemize}[left=0pt, labelsep=1em, itemsep=0pt, wide=0pt, align=parleft]
    \item {\em The maximum independent set problem:} Given a graph $G=(V,E)$ and a number $\ell$, the problem is to determine if $G$ contains an {\em independent set} of size at least $\ell$ i.e. whether there exists  $S\subseteq V, |S|\geq \ell$ such that for any $u,v\in S$, $(u,v)\notin E$. 
    \item We create an instance of $\cProblemU$ by creating one item for each $v\in V$, and two platforms $a,b$ that have edges from all the items. All the edges from the items to the platform $a$ have utility $1$, and the edges to platform $b$ have utility $0$.
    We create a set of groups such that there is a group corresponding to each edge $(u,v)\in E$, consisting of the items that correspond to its end-points. The cost function for the groups is $0$ for platform $b$, and we have the following convex cost function for each group $g$ for platform $a$:
    $$f_a^g(x) = \left\{
\begin{array}{ll}
      0 &\text{ if $x \leq 1$} \\
      x &\text{if $x > 1$} \\
\end{array} 
\right\}$$
\item The platform costs are set as follows:
$$f_a(x) = \left\{
\begin{array}{ll}
      0 &\text{ if $x \leq \ell$} \\
      x &\text{if $x > \ell$} \\
\end{array} 
\right\}$$
$$f_b(x) = \left\{
\begin{array}{ll}
      0 &\text{ if $x \leq |V|-\ell$} \\
      x &\text{if $x > |V| - \ell$} \\
\end{array} 
\right\}$$

\end{itemize} 
Clearly, if $G$ has an independent set $S$ of size at least $\ell$, then we can match the $\ell$ items corresponding to vertices in $S$ to platform $a$, and the remaining items to platform $b$, yielding a matching of utility $\ell$ and total cost $0$. This is because group-related costs are non-zero only for platform $a$, and no two vertices in $S$ share an edge in $G$, so no group incurs a cost. Conversely, suppose there exists a matching with utility $\ell$ and cost $0$. Since each item has utility $1$ on the edges to only platform $a$, there must be exactly $\ell$ items matched to $a$. Group costs are non-trivial only for platform $a$, therefore, $0$ cost means no group contributes to the cost, which implies no two items from the same group are matched to platform $a$. Thus, the $\ell$ items matched to platform $a$ must correspond to an independent set in $G$ of size at least $\ell$, with the remaining items matched to platform $b$.

Since any non-zero cost in the matching corresponds to a violation of the independent set condition, a matching of cost greater than zero implies that $G$ does not contain an independent set of size at least $\ell$. This gives us the inapproximability result.

\end{proof}

\section{Limitations and Future Work}
Our model currently assumes disjoint or laminar group structures, which capture many practical scenarios. However, extending to arbitrary overlapping groups may require significant restructuring of the network construction or may not fit into this framework at all. We also assume that the utility of assigning an item to a platform is known, which holds in some applications. In settings like federated learning, where utility is revealed post-training, predictive models can be used to estimate utility, with updates incorporated after each round.


\section{Experiments}\label{sec:exp}
This section evaluates \Cref{alg:main} on a real-world dataset, the movielens $100k$ dataset \footnote{ https://grouplens.org/datasets/movielens/100k/} comprising $100,000$ user ratings. 
We assess the practical performance and efficiency of our algorithm and compare it against a natural greedy baseline. Our evaluation focuses on the cost-approximation quality of the integral solution returned by our algorithm, relative to both the greedy output and the optimal value of \Cref{LP_int}. Empirically, \Cref{alg:main} consistently outperforms the greedy algorithm and achieves a cost that is only marginally higher than the Linear Programming (LP) optimum. Since the LP value serves as a lower bound on the true optimum ($OPT$), this suggests that our algorithm is near-optimal in practice.

We also examine the effect of running \Cref{alg:main} with all cost functions set to $-\log(x)$, which encourages uniform allocation (similar to the Nash Social Welfare objective). In this setting, we observe that the distribution of items across platforms is nearly uniform. A plot of the variance in the number of items assigned to each platform confirms this behavior, showing a nearly constant line as seen in figures \ref{fig:top_10}, \ref{fig:top_20}, \ref{fig:top_50}, and \ref{fig:top_75}. We also present a group-wise distribution of matched items in each platform in figures \ref{fig:top_10_m}, \ref{fig:top_20_m}, \ref{fig:top_50_m}, and \ref{fig:top_75_m}. These distributions closely mirror the group proportions in the neighborhoods of the corresponding platforms in the input bipartite graph. The group-wise distribution of items in each platform in the input bipartite graph is presented in figures \ref{fig:top_10_m}, \ref{fig:top_20_m}, \ref{fig:top_50_m}, and \ref{fig:top_75_m}.

\subsection{Dataset}\label{dataset}
The movielens $100k$ dataset \cite{dataset} (Harper et al.) consists of $100,000$ ratings (on a scale from $1$ to $5$) from $943$ users for $1682$ movies. In addition to rating data, the dataset includes demographic information about the users, such as age, gender, occupation, and zip code.

In the context of organizing a movie test screening, the goal would be to select an audience that shares an interest in a particular genre while also ensuring demographic diversity. This would enable collecting feedback that reflects the perspectives of various demographic groups, leading to more representative insights. We use user ratings on older movies of similar genres available in the dataset as a proxy for estimating which users are more likely to enjoy the movie being screened.

In our model, users correspond to items and movies to platforms. Each user rating is treated as the utility of the edge between the user (item) and the movie (platform) they rated. We partition users into groups based on age brackets: $15-29$ (group $1$), $30-44$ (group $2$), $45-59$ (group $3$), $60-74$ (group $4$) and $<15$ (group $5$). To promote both diversity of age groups within each movie screening and equity in the number of users matched to each movie, we employ convex cost functions, specifically, $x^2$.

We evaluate \Cref{alg:main} on the top $10$, $20$, $50$, $75$, and $100$ most watched movies from the dataset. We use the number of ratings as a proxy for the number of views for each movie.

\begin{algorithm}[t]
  \caption{Greedy Algorithm}
  \label{alg:gsc}
  \KwIn{$\cI$}
  \KwOut{A matching, $\cM$, with utility at least $\ell$}
  \BlankLine
  $\cM = \emptyset$, $\text{util} = 0$ \\
  \For{$p\in P$}{
    $\sigma_p = 0$ \\
    \For {$j \in g(p)$}{
      $\nu_p^j = 0$
    }
  }
  \While{$\text{util} < \ell$}{
    $\rho^* = 0$, $(i^*, p^*) = \emptyset$, $j^* = 0$ \\
    \For{$(i, p) \in E \setminus \cM$ \text{ such that }$i \in \cItems_j$}{
        $\rho  = \frac{u_{ip}}{\cCF{\sigma_p + 1} - \cCF{\sigma_p} +  \cGCF{j}{\nu_p^j+1} - \cGCF{j}{\nu_p^j}}$ \\
        \If{$\rho > \rho^*$}{
            $\rho^* = \rho$, $(i^*, p^*) = (i,p)$, $j^* = j$
        }}
    $\cM = \cM \cup (i^*, p^*)$\\
    $\text{util} = \text{util} + u_{\cEdge{i^*}{p^*}}$ \\
    $\sigma_{p^*} = \sigma_{p^*}+1$, $\nu_{p^*}^{j^*} = \nu_{p^*}^{j^*}+1$
    }
  \Return $\cM$
\end{algorithm}

\subsection{Experimental Setup and Results}\label{exp_setup_results}
We implement our algorithm in Python $3.10$, and all the experiments are run using Google Colab notebook on a virtual machine with Intel(R) Xeon(R) CPU $@$ $2.20$GHz and $13$GB RAM. In tables \ref{tab:top10}, \ref{tab:top20}, \ref{tab:top50_glop}, \ref{tab:top75_glop}, and \ref{tab:top100_glop}, we report a comparison of the total cost (rounded to the nearest integer) incurred by our algorithm with that of the optimal solution to \Cref{LP_int}, the greedy algorithm (\Cref{alg:gsc}), and a naive greedy baseline. The convex cost function used for every platform and every group-platform pair is $x^2$. These comparisons are conducted under different utility threshold values: $500$, $1000$, and $1500$. Additionally, we report the execution times of \Cref{alg:gsc} and \Cref{alg:main}. The value labeled `Size' in each table caption refers to the total number of edges in the input bipartite graph.

The naive greedy, which selects the highest utility edge iteratively until the utility threshold is met, performs the worst in terms of cost, as expected. A more meaningful comparison is between \Cref{alg:gsc} and \Cref{alg:main}. 
The greedy algorithm \Cref{alg:gsc} selects edges based on the highest utility-to-cost ratio (see \Cref{alg:gsc} for pseudocode). As anticipated, the run time of \Cref{alg:gsc} is lower than that of \Cref{alg:main}, since our algorithm requires constructing a flow network and solving an LP. 
However, our algorithm consistently achieves lower total cost across all runs, except for the first two utility thresholds in \Cref{tab:top20}. Notably, the maximum deviation from the optimal cost returned by \Cref{LP_int} is $26$ units for \Cref{alg:main}, compared to $124$ units for \Cref{alg:gsc}. 

In figures \ref{fig:top_10}, \ref{fig:top_20}, \ref{fig:top_50}, and \ref{fig:top_75}, we observe an almost uniform distribution of items across platforms when all the cost functions are set to $-\log(x)$ for the top $10$, $20$, $50$, and $75$ movies datasets respectively. This is expected, as the resulting objective resembles a Nash Social Welfare kind of objective discussed in \Cref{prop:NSW}. 
We also plot the group-wise distribution of matched items for each platform and observe that they closely reflect the group proportions in the neighborhoods of the corresponding platforms in the input bipartite graph. The group-wise distributions in the input bipartite graph are shown in \Cref{fig:top_10_g} (top 10 movies), \Cref{fig:top_20_g} (top 20 movies), \Cref{fig:top_50_g} (top 50 movies), and \Cref{fig:top_75_g} (top 75 movies), while the group-wise distributions of matched items are shown in \Cref{fig:top_10_m}, \Cref{fig:top_20_m}, \Cref{fig:top_50_m}, and \Cref{fig:top_75_m}.

\begin{table}[ht!]
  \caption{Top 10 most-watched movies \\ Size : 4863} 
  \label{tab:top10}
  \centering
  \small
    \begin{tabular}{l l l l l l l}
      \toprule
      \multirow{2}{*}{\shortstack[l]{Utility \\  Threshold}}
        & \multicolumn{4}{c}{Cost} 
        & \multicolumn{2}{c}{Run-time (s)}\\  
        \cmidrule(r){2-5} \cmidrule(l){6-7} 
        & Naive Greedy 
        & \Cref{alg:gsc}
        & \Cref{LP_int}
        & \Cref{alg:main}
        &\Cref{alg:gsc}
        & \Cref{alg:main} \\
      \midrule
      500  & 3384    
           & 1816 
           & 1788 
           & 1802  
           & 0.757
           & 1.545 \\
      1000 & 11126
           & 6810
           & 6768
           & 6794 
           & 0.77
           & 1.306 \\
      1500 & 23992 
           & 15224 
           & 15100
           & 15100
           & 1.111
           & 1.331 \\
      \bottomrule
    \end{tabular}%
\end{table}

\begin{table}[ht!]
  \caption{Top 20 most-watched movies \\ Size: 8716} 
  \label{tab:top20}
  \small
  \centering
    \begin{tabular}{l l l l l l l}
      \toprule
      \multirow{2}{*}{\shortstack[l]{Utility \\  Threshold}}
        & \multicolumn{4}{c}{Cost} 
        & \multicolumn{2}{c}{Run-time (s)}\\  
       \cmidrule(r){2-5} \cmidrule(l){6-7} 
        & Naive Greedy 
        & \Cref{alg:gsc}
        & \Cref{LP_int}
        & \Cref{alg:main}
        &\Cref{alg:gsc}
        & \Cref{alg:main} \\
      \midrule
      500   & 1698  
          & 1076 
           & 1072 
           & 1088  
           & 0.736
           & 2.169 \\
      1000  & 5780  
            & 3698 
           & 3684
           & 3704 
           & 1.698 
           & 3.163 \\
      1500  & 11834 
            & 8030 
           & 7965
           & 7986 
           & 1.902
           & 2.335 \\
      \bottomrule
    \end{tabular}%
\end{table}

\begin{table}[ht!]
  \caption{Top 50 most-watched movies \\ Size: 17841} 
  \label{tab:top50_glop}
  \small
  \centering
    \begin{tabular}{l l l l l l l}
      \toprule
        \multirow{2}{*}{\shortstack[l]{Utility \\  Threshold}}
        & \multicolumn{4}{c}{Cost} 
        & \multicolumn{2}{c}{Run-time (s)}\\  
        \cmidrule(r){2-5} \cmidrule(l){6-7} 
        & Naive Greedy 
        & \Cref{alg:gsc}
        & \Cref{LP_int}
        & \Cref{alg:main}
        &\Cref{alg:gsc}
        & \Cref{alg:main} \\

      \midrule
      500    & 1070  
             & 700    
           & 700  
           & 700  
           & 1.499
           & 6.443 \\
      1000  & 3072  
            & 1890 
           & 1870
           & 1876 
           & 3.873
           & 5.185 \\
      1500  & 6292  
            & 3796 
           & 3772
           & 3772
           & 4.813
           & 5.854 \\
      \bottomrule
    \end{tabular}%
\end{table}

\begin{table}[ht!]
  \caption{Top 75 most-watched movies \\ Size: 24251} 
  \label{tab:top75_glop}
  \small
  \centering
    \begin{tabular}{l l l l l l l}
      \toprule
      \multirow{2}{*}{\shortstack[l]{Utility \\  Threshold}}
        & \multicolumn{4}{c}{Cost} 
        & \multicolumn{2}{c}{Run-time (s)}\\  
        \cmidrule(r){2-5} \cmidrule(l){6-7} 
        & Naive Greedy 
        & \Cref{alg:gsc}
        & \Cref{LP_int}
        & \Cref{alg:main}
        &\Cref{alg:gsc}
        & \Cref{alg:main} \\
      \midrule
      500    & 920  
             & 650   
           & 650
           & 650 
           & 2.331
           & 9.263 \\
      1000  & 2498  
            & 1550
           & 1550
           & 1550
           & 4.345
           & 9.648 \\
      1500  & 4840  
            & 2882
           & 2857
           & 2868
           & 6.291
           & 9.866 \\
      \bottomrule
    \end{tabular}%
\end{table}

\begin{table}[ht!]
  \caption{Top 100 most-watched movies \\ Size: 29931} 
  \label{tab:top100_glop}
  \small
  \centering
    \begin{tabular}{l l l l l l l}
      \toprule
      \multirow{2}{*}{\shortstack[l]{Utility \\  Threshold}}
        & \multicolumn{4}{c}{Cost} 
        & \multicolumn{2}{c}{Run-time (s)}\\  
        \cmidrule(r){2-5} \cmidrule(l){6-7} 
        & Naive Greedy 
        & \Cref{alg:gsc}
        & \Cref{LP_int}
        & \Cref{alg:main}
        &\Cref{alg:gsc}
        & \Cref{alg:main}\\
      \midrule
      500    & 840  
             & 600    
           & 600  
           & 600  
           & 2.877
           & 12.601 \\
      1000  & 2230  
            & 1400
           & 1400
           & 1400
           & 6.644
           & 12.742 \\
      1500  & 4208 
            & 2412
           & 2402
           & 2402
           & 7.571
           & 12.910 \\
      \bottomrule
    \end{tabular}%
\end{table}

\begin{figure}[ht!]
\centering
\includegraphics[scale=0.7]{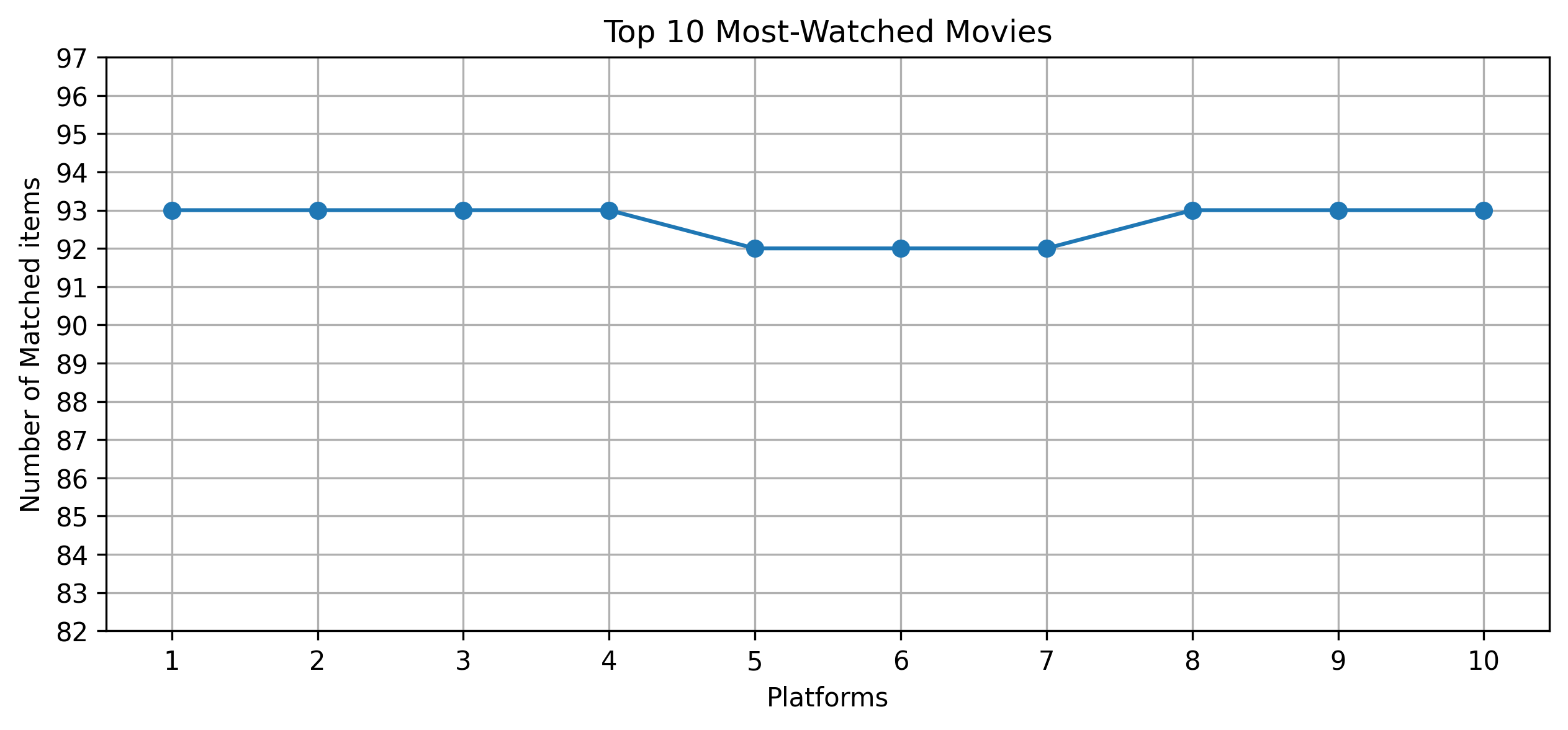}
\caption{Variance of the total number of items matched to platforms in the Top 10 most watched movies dataset}
\label{fig:top_10}
\end{figure}

\begin{figure}[ht!]
\centering
\includegraphics[scale=0.7]{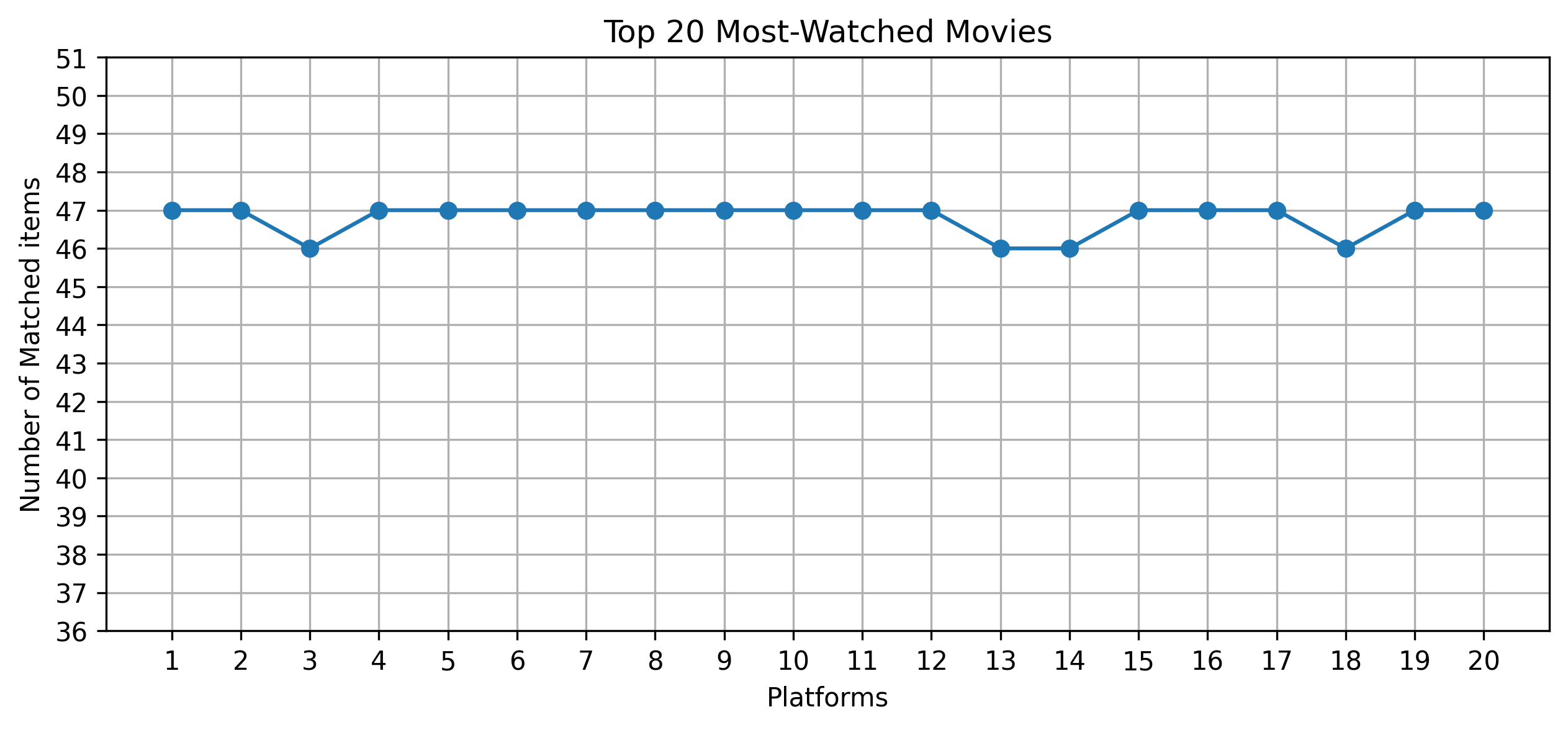}
\caption{Variance of the total number of items matched to platforms in the Top 20 most watched movies dataset}
\label{fig:top_20}
\end{figure}

\begin{figure}[ht!]
\centering
\includegraphics[scale=0.47]{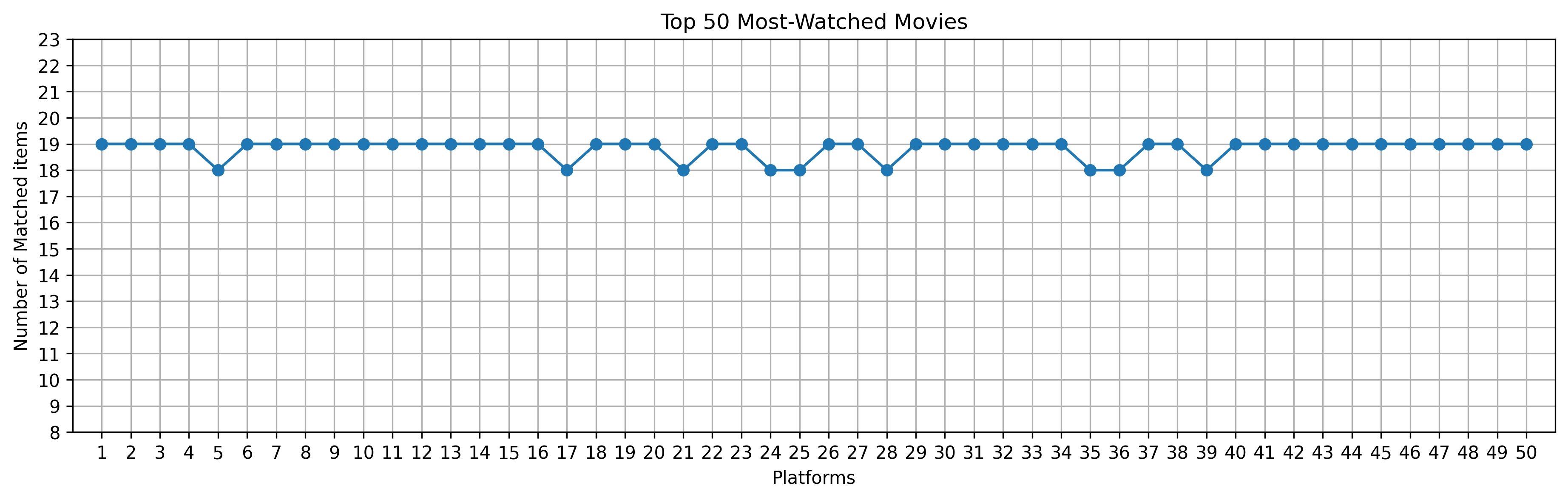}
\caption{Variance of the total number of items matched to platforms in the Top 50 most watched movies dataset}
\label{fig:top_50}
\end{figure}

\begin{figure}[ht!]
\centering
\includegraphics[scale=0.35]{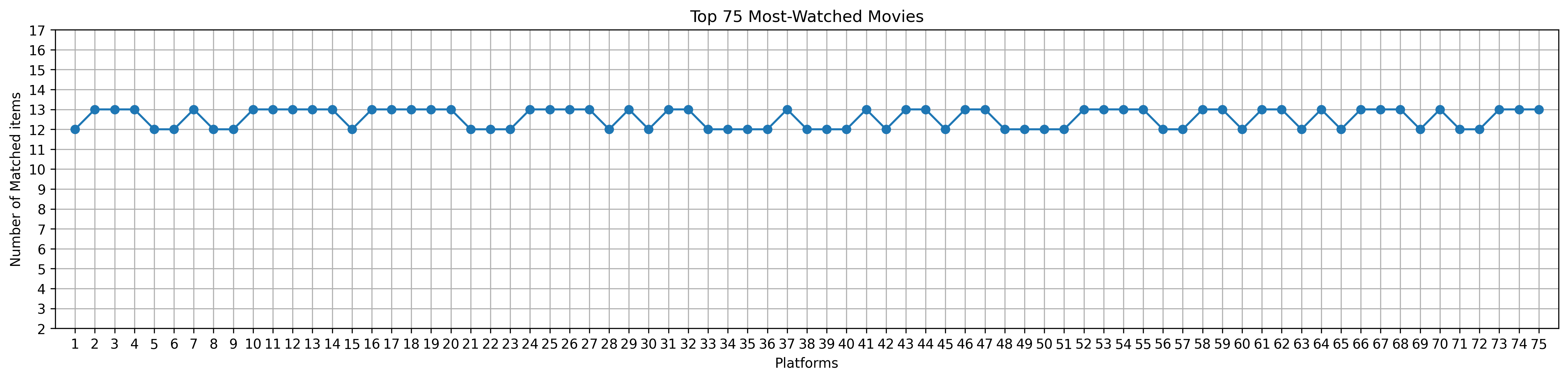}
\caption{Variance of the total number of items matched to platforms in the Top 75 most watched movies dataset}
\label{fig:top_75}
\end{figure}

\begin{figure}[ht!]
\centering
\includegraphics[scale=0.65]{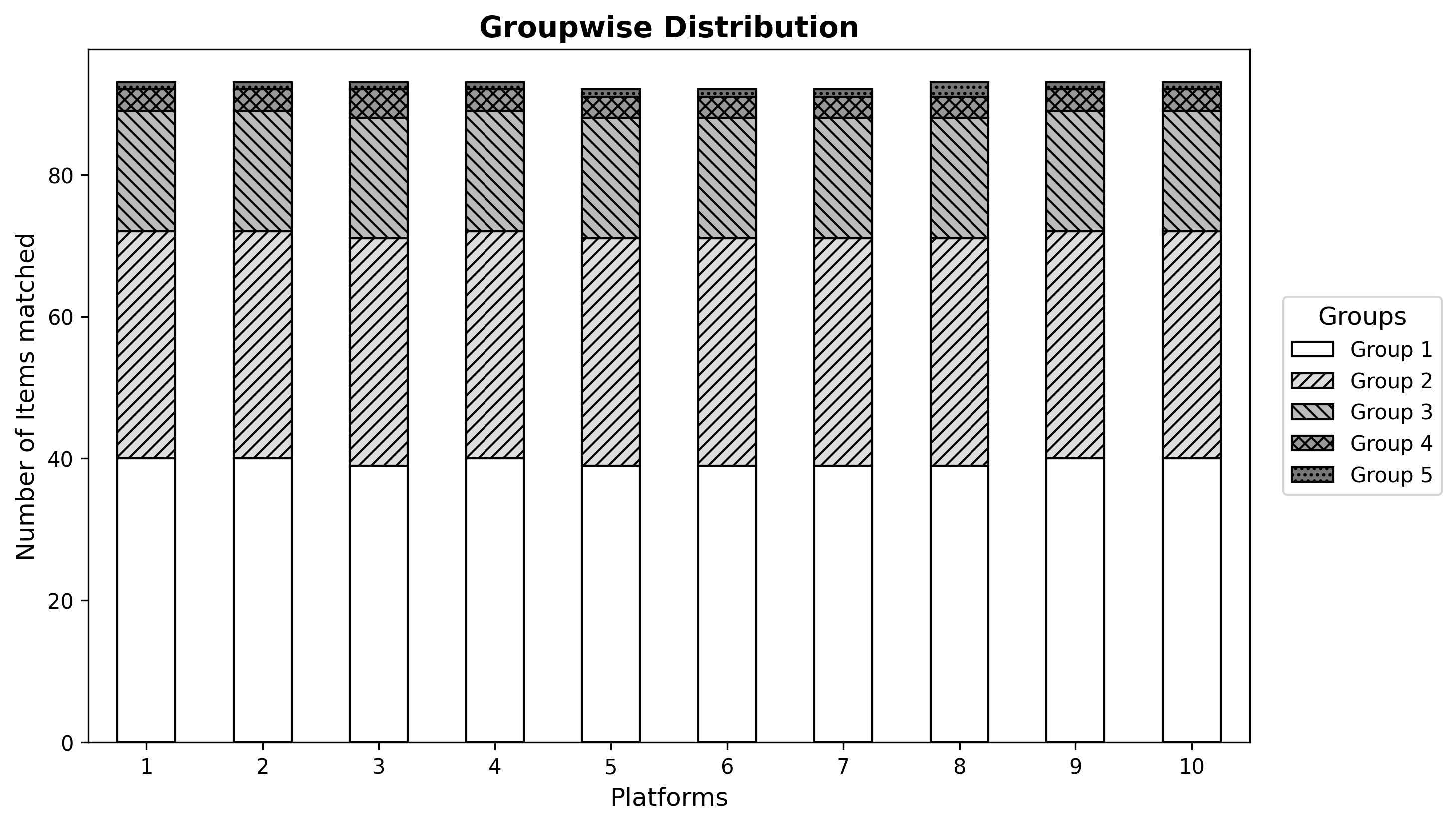}
\caption{Top $10$ most-watched movies}
\label{fig:top_10_m}
\end{figure}
\begin{figure}[ht!]
\centering
\includegraphics[scale=0.55]{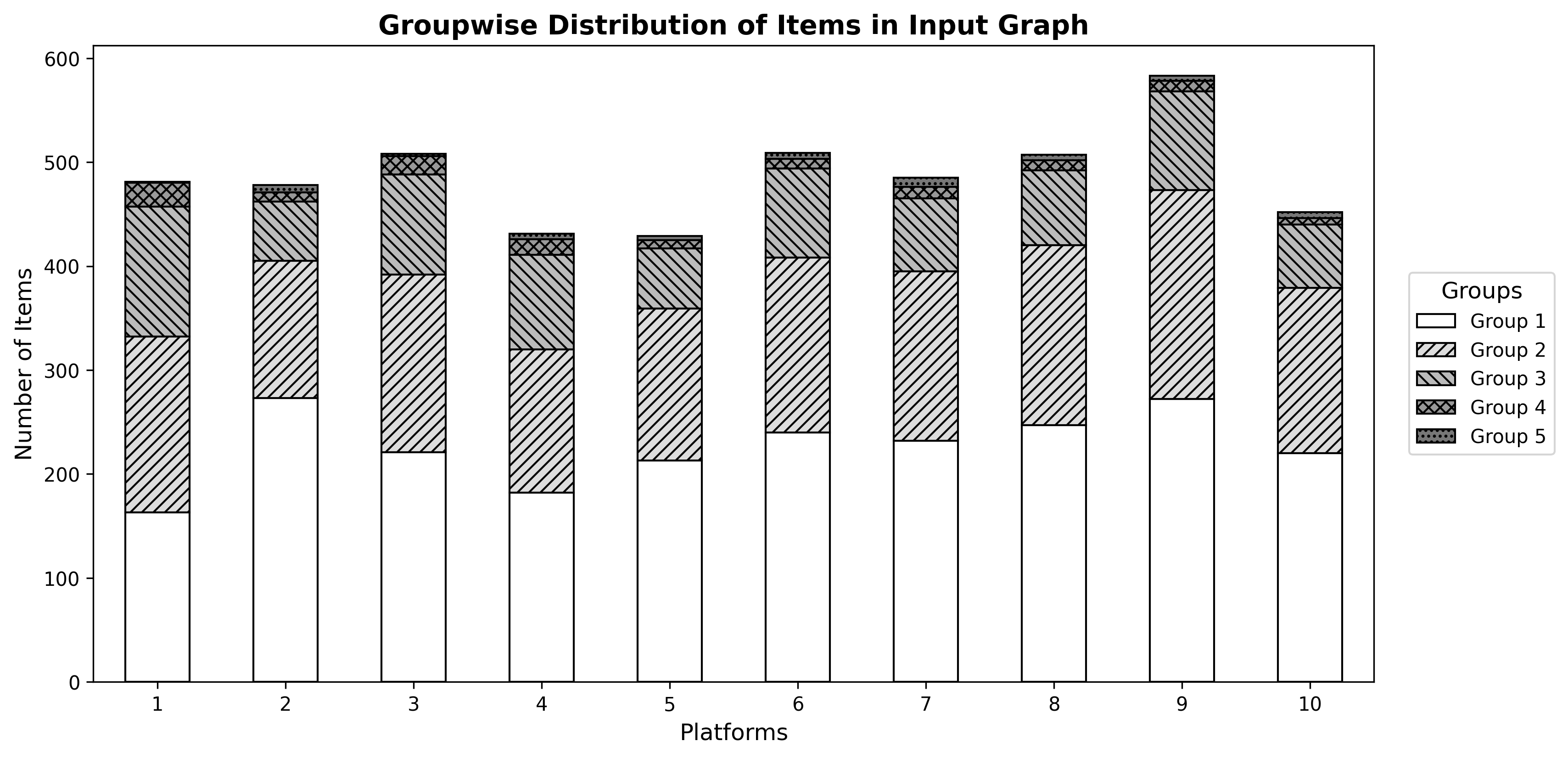}
\caption{Top $10$ most-watched movies}
\label{fig:top_10_g}
\end{figure}

\begin{figure}[ht!]
\centering
\includegraphics[scale=0.65]{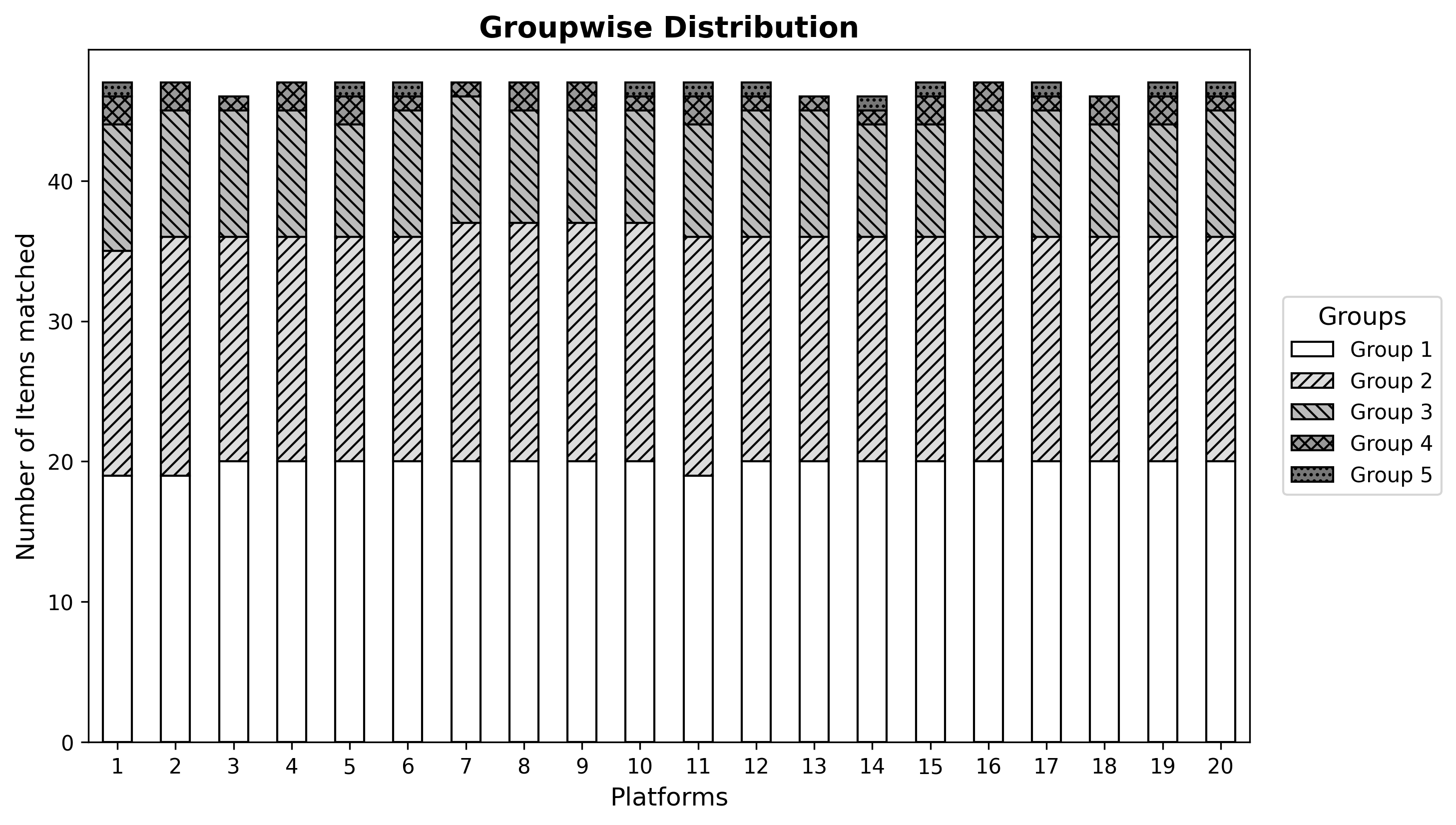}
\caption{Top $20$ most-watched movies}
\label{fig:top_20_m}
\end{figure}
\begin{figure}[ht!]
\centering
\includegraphics[scale=0.55]{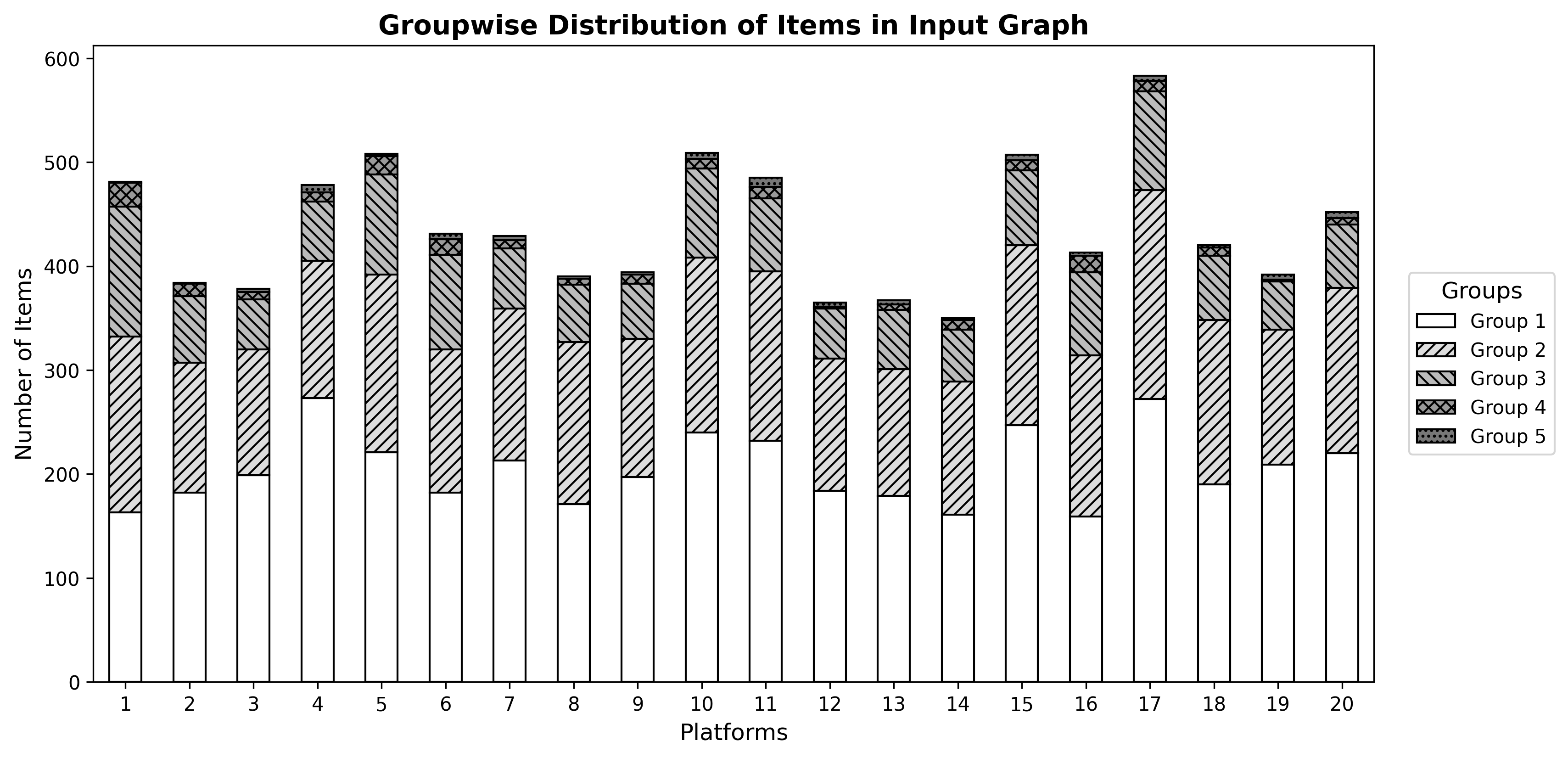}
\caption{Top $20$ most-watched movies}
\label{fig:top_20_g}
\end{figure}

\begin{figure}[ht!]
\centering
\includegraphics[scale=0.38]{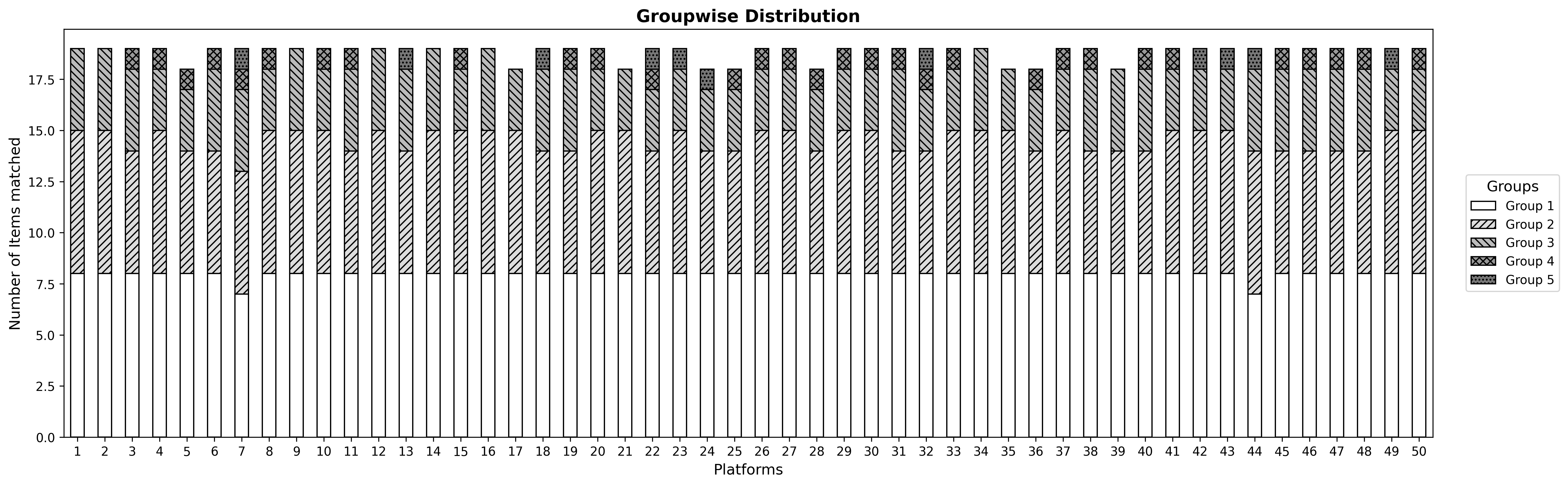}
\caption{Top $50$ most-watched movies}
\label{fig:top_50_m}
\end{figure}
\begin{figure}[ht!]
\centering
\includegraphics[scale=0.38]{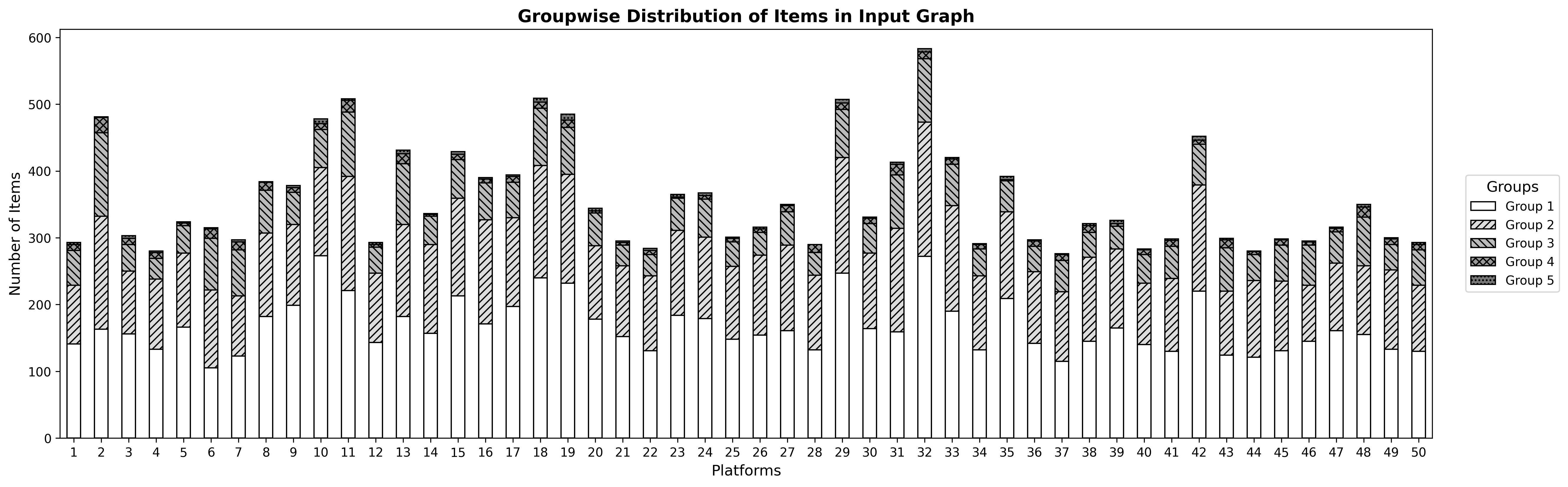}
\caption{Top $50$ most-watched movies}
\label{fig:top_50_g}
\end{figure}

\begin{figure}[ht!]
\centering
\includegraphics[scale=0.37]{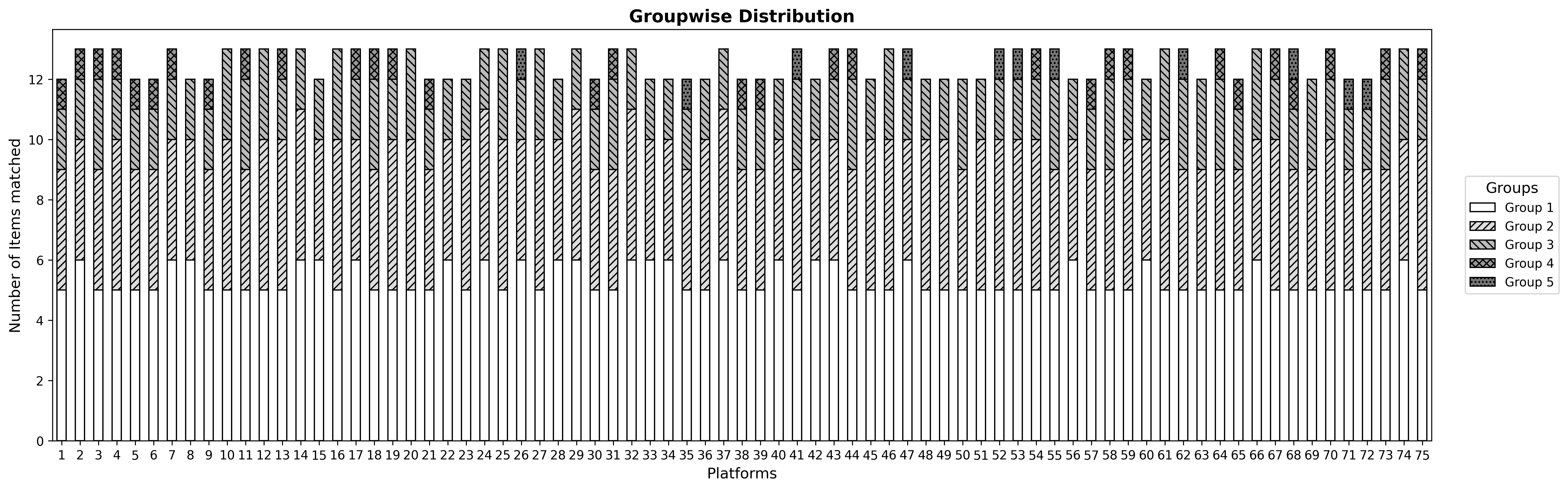}
\caption{Top $75$ most-watched movies}
\label{fig:top_75_m}
\end{figure}
\begin{figure}[ht!]
\centering
\includegraphics[scale=0.37]{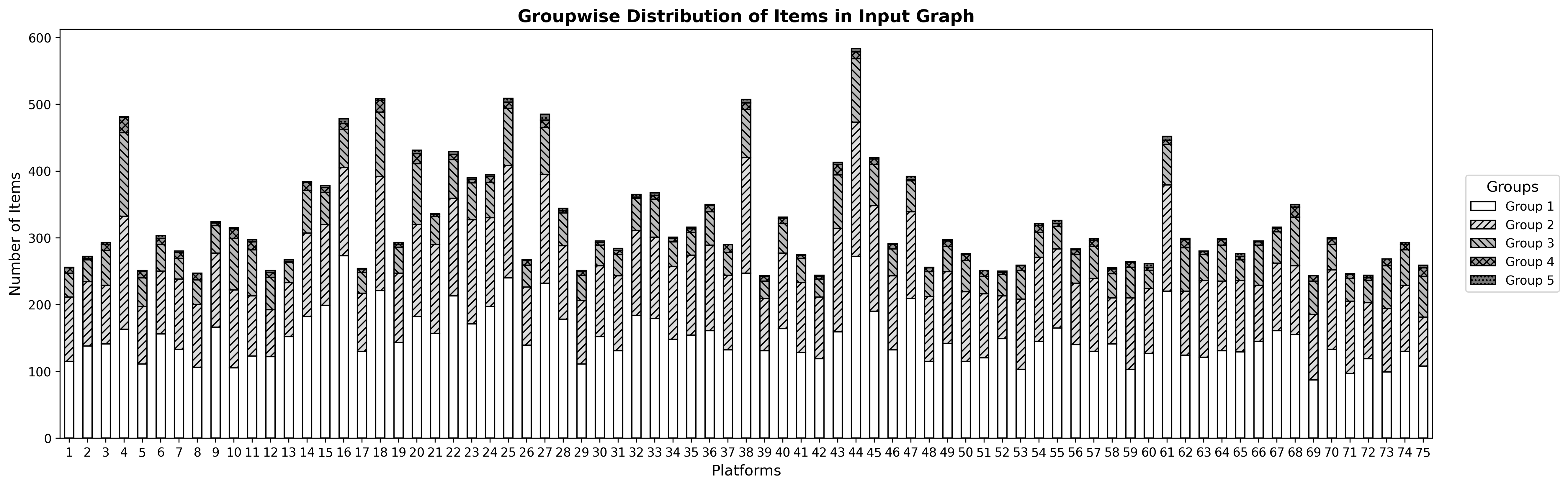}
\caption{Top $75$ most-watched movies}
\label{fig:top_75_g}
\end{figure}

\section{Uniform utilities}\label{sec:CCM}
This is a special case of the $\cProblemU$ problem defined in \Cref{sec:prelims} where all the edges have the same utility. Given an instance of this problem, $\cI = (\cGraph, \cG_1, \cG_2, \cdots, \cG_{\tau}, \ell, q, F)$, where $q$ is the utility in every edge, the problem reduces to finding a Min-Cost $\frac{\ell}{q}$-Flow on the network constructed using the steps in \Cref{subsec:Flow} or \Cref{sec:laminar} based on if the groups are disjoint or laminar. We use \Cref{alg:matchingSize} to compute such a matching.
We first restate \Cref{thm:CCM} then use \Cref{alg:main} to provide its proof.

\begin{algorithm}[t]
\caption{Min-Cost $\frac{\ell}{q}$-Flow($\cI$)}
\label{alg:matchingSize}
\nonl \textbf{Input} :  $\cI = (\cGraph, \cG_1, \cG_2, \cdots, \cG_{\tau}, \ell, q, F)$ \\
\nonl \textbf{Output} : An matching, $\cM$ such that $|\cM| \geq \frac{\ell}{q}$\\
$\cM = \phi$\\
Follow the steps in \Cref{subsec:Flow} or \Cref{sec:laminar} to construct the flow network $\cFN_U$ \\
Run the Shortest Successive Paths algorithm on $\cFN_U$ for a flow value of $\frac{\ell}{q}$ and store the flow in $FL$\\ \label{step:computeFlow}
$FE = \{\cEdge{\cItem}{w}: \cItem \in \cItems \text{ and } \cEdge{\cItem}{w} \text{ is part of the flow } FL\}$ \\
\For{$\cEdge{\cItem}{w} \in FE$}{ \label{step:Forloop}
    Let $p$ be the platform such that $w$ is some 
    copy of $p$ in the group layer.\\
    $\cM = \cM \cup \cEdge{\cItem}{p}$ \label{step:endLoop}
}
Return $\cM$
\end{algorithm}

\CCM*
\begin{proof}
    The Successive Shortest Path (SSP) algorithm is a classical polynomial-time method for solving the minimum-cost flow problem in network flow theory. It incrementally augments flow along the shortest (minimum-cost) paths until the desired flow value is reached.  For further details on the algorithm, readers may consult any standard reference, such as \cite{AhujaMO93}. From the Min Cost Max Flow integrality theorem \cite{AhujaMO93}, we know that when the capacities and costs in the flow network are integral, there is always an integral solution to the problem, therefore, when $\frac{\ell}{q}$ is an integer, \Cref{alg:matchingSize} returns a valid Matching with $0$ cost violation.
    The run-time complexity of the SSP algorithm is $O(|V_{\cFN_U}|^2|E_{\cFN_U}|B)$ where $V_{\cFN_U}$ is the set of nodes, $E_{\cFN_U}$ is the set of edges and $B$ is the upper bound on the largest capacity of any node. By \Cref{lem:runtime}, in our flow network, either $|V_{\cFN_U}| \leq n+2m+|E|+2$, $|E_{\cFN_U}| = n + m + 3|E|$ or $|E_{\cFN_U}| \leq n + m + (d+2)|E|$ (\Cref{lem:lam_runtime}) and $B = \max_{p \in P}\Delta(p)$, therefore, using the SSP algorithm, \Cref{alg:main} can run in $O((n+2m+|E|+2)^2(n + m + (d+2)|E|)(\max_{p \in P}\Delta(p)))$.
\end{proof}

\bibliographystyle{unsrt}
\bibliography{reference}

\begin{thebibliography}{10}

\bibitem{vishnoi_fairness}
L.~Elisa Celis, Damian Straszak, and Nisheeth~K. Vishnoi.
\newblock Ranking with fairness constraints.
\newblock In {\em ICALP}, 2017.

\bibitem{luss_leximin_fair}
Hanan Luss.
\newblock On equitable resource allocation problems: A lexicographic minimax approach.
\newblock {\em Operations Research}, 1999.

\bibitem{devanur_ranking}
Nikhil~R. Devanur, Kamal Jain, and Robert~D. Kleinberg.
\newblock Randomized primal-dual analysis of ranking for online bipartite matching.
\newblock In {\em SODA '13}, 2013.

\bibitem{CHRG16}
Matthew Costello, James Hawdon, Thomas Ratliff, and Tyler Grantham.
\newblock Who views online extremism? individual attributes leading to exposure.
\newblock {\em Comput. Hum. Behav.}, 2016.

\bibitem{halevi_fair_allocation}
Erel Segal-Halevi and Warut Suksompong.
\newblock Democratic fair allocation of indivisible goods.
\newblock {\em Artificial Intelligence}, 2019.

\bibitem{KMM15}
Matthew Kay, Cynthia Matuszek, and Sean~A. Munson.
\newblock Unequal representation and gender stereotypes in image search results for occupations.
\newblock In {\em Proceedings of the 33rd Annual ACM Conference on Human Factors in Computing Systems}, CHI ’15, 2015.

\bibitem{BCZSK16}
Tolga Bolukbasi, Kai{-}Wei Chang, James~Y. Zou, Venkatesh Saligrama, and Adam~Tauman Kalai.
\newblock Man is to computer programmer as woman is to homemaker? debiasing word embeddings.
\newblock In Daniel~D. Lee, Masashi Sugiyama, Ulrike von Luxburg, Isabelle Guyon, and Roman Garnett, editors, {\em Advances in Neural Information Processing Systems 29: Annual Conference on Neural Information Processing Systems 2016, December 5-10, 2016, Barcelona, Spain}, pages 4349--4357, 2016.

\bibitem{GroupIndividual}
Atasi Panda, Anand Louis, and Prajakta Nimbhorkar.
\newblock Individual fairness under group fairness constraints in bipartite matching - one framework to approximate them all.
\newblock In {\em {IJCAI-24}3}, 2024.

\bibitem{case-study}
{Cowen Institute}.
\newblock Case studies of school choice and open enrollment in four cities.
\newblock Technical report, Cowen Institute, 2011.

\bibitem{fair_clustering}
Suman~Kalyan Bera, Deeparnab Chakrabarty, Nicolas Flores, and Maryam Negahbani.
\newblock Fair algorithms for clustering.
\newblock In {\em NeurIPS}, 2019.

\bibitem{fair_ranking}
L.~Elisa Celis, Damian Straszak, and Nisheeth~K. Vishnoi.
\newblock Ranking with fairness constraints.
\newblock In {\em {ICALP}}, 2018.

\bibitem{fair_multivote}
L.~Elisa Celis, Lingxiao Huang, and Nisheeth~K. Vishnoi.
\newblock Multiwinner voting with fairness constraints.
\newblock In J{\'{e}}r{\^{o}}me Lang, editor, {\em {IJCAI}}, 2018.

\bibitem{one-sidedPreference_softquota}
Santhini~K. A., Raghu~Raman Ravi, and Meghana Nasre.
\newblock Matchings under one-sided preferences with soft quotas.
\newblock In {\em {IJCAI}, Macao, SAR, China}, 2023.

\bibitem{NSWFair}
Ioannis Caragiannis, David Kurokawa, Herv{\'{e}} Moulin, Ariel~D. Procaccia, Nisarg Shah, and Junxing Wang.
\newblock The unreasonable fairness of maximum nash welfare.
\newblock {\em {ACM} Trans. Economics and Comput.}, 7(3):12:1--12:32, 2019.

\bibitem{NSWFairSurvey}
Georgios Amanatidis, Georgios Birmpas, Aris Filos{-}Ratsikas, and Alexandros~A. Voudouris.
\newblock Fair division of indivisible goods: {A} survey.
\newblock In {\em {IJCAI}}, 2022.

\bibitem{GroupFairMatching}
Govind~S. Sankar, Anand Louis, Meghana Nasre, and Prajakta Nimbhorkar.
\newblock Matchings with group fairness constraints: Online and offline algorithms.
\newblock In {\em {IJCAI}}, 2021.

\bibitem{ApproxNSW}
Richard Cole and Vasilis Gkatzelis.
\newblock Approximating the nash social welfare with indivisible items.
\newblock {\em {SIAM} J. Comput.}, 47(3):1211--1236, 2018.

\bibitem{halabian_resourceallocation}
Hassan {Halabian}, Ioannis {Lambadaris}, and Chung-Horng {Lung}.
\newblock Optimal server assignment in multi-server parallel queueing systems with random connectivities and random service failures.
\newblock In {\em IEEE International Conference on Communications (ICC)}, 2012.

\bibitem{KidneyExchange}
Golnoosh Farnadi, William St-Arnaud, Behrouz Babaki, and Margarida Carvalho.
\newblock Individual fairness in kidney exchange programs.
\newblock {\em AAAI}, 35(13):11496--11505, May 2021.

\bibitem{abdulkadiroglu2003}
Atila Abdulkadiroglu and T.~S{\"o}nmez.
\newblock School choice: A mechanism design approach.
\newblock {\em The American Economic Review}, 2003.

\bibitem{CandidateSelection}
Xiaohui Bei, Shengxin Liu, Chung~Keung Poon, and Hongao Wang.
\newblock Candidate selections with proportional fairness constraints.
\newblock In {\em {AAMAS}}, 2020.

\bibitem{SummerInternship}
Haris Aziz, Anton Baychkov, and P{\'{e}}ter Bir{\'{o}}.
\newblock Summer internship matching with funding constraints.
\newblock In {\em {AAMAS}}, 2020.

\bibitem{Hospital-Resident}
Hiromichi Goko, Kazuhisa Makino, Shuichi Miyazaki, and Yu~Yokoi.
\newblock Maximally satisfying lower quotas in the hospitals/residents problem with ties.
\newblock In {\em 39th International Symposium on Theoretical Aspects of Computer Science, {STACS} 2022, March 15-18, 2022, Marseille, France (Virtual Conference)}, volume 219 of {\em LIPIcs}, pages 31:1--31:20. Schloss Dagstuhl - Leibniz-Zentrum f{\"{u}}r Informatik, 2022.

\bibitem{AzizBiroYokoo2022}
Haris Aziz, P{\'{e}}ter Bir{\'{o}}, and Makoto Yokoo.
\newblock Matching market design with constraints.
\newblock In {\em {AAAI} 2022, {IAAI} 2022,{EAAI} 2022}, 2022.

\bibitem{ClassifiedStableMatching}
Chien{-}Chung Huang.
\newblock Classified stable matching.
\newblock In Moses Charikar, editor, {\em Proceedings of the Twenty-First Annual {ACM-SIAM} Symposium on Discrete Algorithms, {SODA} 2010, Austin, Texas, USA, January 17-19, 2010}, pages 1235--1253. {SIAM}, 2010.

\bibitem{IsraeliGapYear}
Yannai~A. Gonczarowski, Noam Nisan, Lior Kovalio, and Assaf Romm.
\newblock Matching for the israeli: Handling rich diversity requirements.
\newblock In Anna Karlin, Nicole Immorlica, and Ramesh Johari, editors, {\em Proceedings of the 2019 {ACM} Conference on Economics and Computation, {EC} 2019, Phoenix, AZ, USA, June 24-28, 2019}, page 321. {ACM}, 2019.

\bibitem{IndiaReservation}
Tayfun Sönmez and M.~Bumin Yenmez.
\newblock Affirmative action in india via vertical, horizontal, and overlapping reservations.
\newblock {\em Econometrica}, 90(3):1143--1176, 2022.

\bibitem{Rank-MaximalAndPOpular}
Meghana Nasre, Prajakta Nimbhorkar, and Nada Pulath.
\newblock Classified rank-maximal matchings and popular matchings -- algorithms and hardness.
\newblock In Ignasi Sau and Dimitrios~M. Thilikos, editors, {\em Graph-Theoretic Concepts in Computer Science}, pages 244--257, Cham, 2019. Springer International Publishing.

\bibitem{DiversityMatching}
Govind~S. Sankar, Anand Louis, Meghana Nasre, and Prajakta Nimbhorkar.
\newblock Matchings under diversity constraints, 2022.

\bibitem{school_choice_expansion}
Federico Bobbio, Margarida Carvalho, Andrea Lodi, Ignacio Rios, and Alfredo Torrico.
\newblock Capacity planning in stable matching: An application to school choice.
\newblock In {\em {EC}}, page 295. {ACM}, 2023.

\bibitem{hopital-resident_expansion}
Kenshi Abe, Junpei Komiyama, and Atsushi Iwasaki.
\newblock Anytime capacity expansion in medical residency match by monte carlo tree search.
\newblock In {\em {IJCAI}}, 2022.

\bibitem{college_admissions_expansion}
Federico Bobbio, Margarida Carvalho, Andrea Lodi, and Alfredo Torrico.
\newblock Capacity expansion in the college admission problem.
\newblock {\em CoRR}, abs/2110.00734, 2021.

\bibitem{Hospital-resident_stable}
Keshav Ranjan, Meghana Nasre, and Prajakta Nimbhorkar.
\newblock Optimal capacity modification for strongly stable matchings.
\newblock {\em CoRR}, 2024.

\bibitem{EHLERS2014648}
Lars Ehlers, Isa~E. Hafalir, M.~Bumin Yenmez, and Muhammed~A. Yildirim.
\newblock School choice with controlled choice constraints: Hard bounds versus soft bounds.
\newblock {\em Journal of Economic Theory}, 153:648--683, 2014.

\bibitem{Yenmez}
Federico Echenique and M.~Bumin Yenmez.
\newblock How to control controlled school choice.
\newblock {\em American Economic Review}, 2015.

\bibitem{ijcai2020softconstraints}
Haris Aziz, Serge Gaspers, and Zhaohong Sun.
\newblock Mechanism design for school choice with soft diversity constraints.
\newblock In {\em {IJCAI}}, 2020.

\bibitem{overlappingSchoolChoice}
Ryoji Kurata, Naoto Hamada, Atsushi Iwasaki, and Makoto Yokoo.
\newblock Controlled school choice with soft bounds and overlapping types.
\newblock {\em J. Artif. Intell. Res.}, 2017.

\bibitem{Aziz_Biró_Yokoo_2022}
Haris Aziz, Péter Biró, and Makoto Yokoo.
\newblock Matching market design with constraints.
\newblock {\em Proceedings of the AAAI Conference on Artificial Intelligence}, 2022.

\bibitem{one-sided_cost-based}
Santhini~K. A., Govind~S. Sankar, and Meghana Nasre.
\newblock Optimal matchings with one-sided preferences: Fixed and cost-based quotas.
\newblock In {\em {AAMAS}, Auckland, New Zealand}, 2022.

\bibitem{two-sided_cost-based}
Girija Limaye and Meghana Nasre.
\newblock Optimal cost-based allocations under two-sided preferences.
\newblock In {\em {IWOCA}}, 2023.

\bibitem{AhujaMO93}
Ravindra~K. Ahuja, Thomas~L. Magnanti, and James~B. Orlin.
\newblock {\em Network flows: theory, algorithms, and applications}.
\newblock Prentice-Hall, Inc., 1993.

\bibitem{Gallo_Flow}
Giorgio Gallo and C.~Sodini.
\newblock Extreme points and adjacency relationship in the flow polytope.
\newblock {\em Calcolo}, 15:277--288, 09 1978.

\bibitem{dataset}
F.~Maxwell Harper and Joseph~A. Konstan.
\newblock The movielens datasets: History and context.
\newblock {\em ACM Trans. Interact. Intell. Syst.}, 5(4), December 2015.

\end{thebibliography}
\newpage
\end{document}